\documentclass[acmlarge]{acmart}

\usepackage{amsmath}    
\usepackage{graphicx}   
\usepackage{epstopdf}
\usepackage{verbatim}   
\usepackage{color}      
\usepackage{subfigure}  
\usepackage{algorithm}
\usepackage{array}
\usepackage{algorithmic}
\usepackage{nomencl}
\usepackage{amssymb}
\usepackage{url}
\newtheorem{myrule}{Rule}

\def\etal{\textit{et al.}}

\def\e2e{end-to-end}

\def\toolname{Y2U}
\def\yakindu{Yakindu}
\def\uppaal{UPPAAL}

\def\toolweb{\url{www.cs.iit.edu/~code/software/Y2U}}

\def\clockSet{C}
\def\clock{c}
\def\actionSet{A}
\def\action{a}
\def\varSet{V}
\def\var{v}
\def\invariant{I}
\def\guardSet{G}
\def\guard{g}
\def\stateSet{S}
\def\state{s}
\def\automata{U}
\def\automataSet{\mathcal{U}}
\def\tranSet{T}
\def\tran{t}
\def\tranPrioritySet{\Gamma}
\def\tranPriority{\gamma}
\def\statechartPriority{\rho}
\def\sysstateSet{\Theta}
\def\sysstate{\theta}
\def\statechart{Y}
\def\statechartSet{\mathcal{Y}}
\def\trace{\mathcal{T}}
\def\prifun{\Phi}
\def\step{\alpha}
\newcommand{\valueFun}[1]{\nu(#1)}
\newcommand{\indexValueFun}[2]{\nu_{#1}(#2)}
\def\event{\epsilon}
\def\inStateActionAssign{\Delta_{\mathtt{in}}}
\def\outStateActionAssign{\Delta_{\mathtt{out}}}
\def\outAction{\action_{\mathtt{out}}}
\def\inAction{\action_{\mathtt{in}}}
\def\tranAction{\action_{\mathtt{tran}}}
\def\eventVarSet{\varSet_{\event}}
\def\increaseStep{\mathtt{Inc}(\step)}
\def\channel{\mathtt{chan}}
\def\outTran{\tran^{\mathtt{out}}}
\def\outTranSet{\tranSet^{\mathtt{out}}}
\newcommand{\automataSetStep}[1]{\automataSet^{(#1)}} 
\def\y{\mathtt{y}}
\def\u{\mathtt{u}}
\def\timer{\tau}
\def\inState{\state^{\mathtt{in}}}
\def\outState{\state^{\mathtt{out}}}
\def\timerAutomata{\automata_{\timer}}
\def\timerAutomataSet{\automataSet_{\timer}}
\def\eventAutomata{\automata_{\event}}
\def\eventAutomataSet{\automataSet_{\event}}
\def\highPriTranSet{\tranSet^{\mathtt{hp}}}
\def\highPriTran{\tran^{\mathtt{hp}}}
\def\labelSet{\Sigma}
\def\labels{\sigma}

\renewcommand\footnotetextcopyrightpermission[1]{} 
\begin{document}

\title{Formalism for Supporting the Development of Verifiably Safe Medical Guidelines with Statecharts}
\subtitle{(Technical Report)}

\author{Chunhui Guo}
\email{cguo13@hawk.iit.edu}
\author{Zhicheng Fu}
\email{zfu11@hawk.iit.edu}
\affiliation{
	\institution{Illinois Institute of Technology}
	\city{Chicago}
	\postcode{60616}
}

\author{Zhenyu Zhang}
\email{zzhang4430@sdsu.edu}
\author{Shangping Ren}
\email{sren@sdsu.edu}
\affiliation{
	\institution{San Diego State University}
	\city{San Diego}
	\postcode{92182}
}

\author{Lui Sha}
\email{lrs@illinois.edu}
\affiliation{
	\institution{University of Illinois at Urbana-Champaign}
	\city{Urbana}
	\postcode{61801}
}

\begin{abstract}
	Improving the effectiveness and safety of patient care is the ultimate objective for medical cyber-physical systems. Many medical best practice guidelines exist, but most of the existing guidelines in handbooks are difficult for medical staff to remember and apply clinically. Furthermore, although the guidelines have gone through clinical validations, validations by medical professionals alone do not provide guarantees for the safety of medical cyber-physical systems. Hence, formal verification is also needed. The paper presents the formal semantics for a framework that we developed to support the development of verifiably safe medical guidelines.

The framework allows computer scientists to work together with medical professionals to transform medical best practice guidelines into executable statechart models, \yakindu\ in particular, so that medical functionalities and properties can be quickly prototyped and validated. Existing formal verification technologies, \uppaal\ timed automata in particular, is integrated into the framework to provide formal verification capabilities to verify safety properties. However, some components used/built into the framework, such as the open-source \yakindu\ statecharts as well as the transformation rules from statecharts to timed automata, do not have built-in semantics. The ambiguity becomes unavoidable unless formal semantics is defined for the framework, which is what the paper is to present.

%
%
%

\end{abstract}

\maketitle
\fancyfoot{}
\thispagestyle{empty}

\section{Introduction}
\label{sec:intro}
Medical best practice guidelines are systematically
developed statements that intend to assist clinicians in
making decisions about appropriate health care procedures in specific circumstances~\cite{NAP1626}.
They aim to improve the quality of patient care by
encouraging interventions of proven benefit and discouraging the use of
ineffective or potentially harmful interventions~\cite{Woolf1999BMJ}.
The library of the Guidelines International Network has 6,187 documents
from 76 countries and the National Guideline Clearinghouse in
the United States has 2,017 guideline summaries~\cite{GuidelineLib}. 
However, most of these existing guidelines
are lengthy and difficult for medical professionals to
remember and apply clinically.

On the other hand, a study shows that the patients' death rate can be significantly reduced by
computerizing medical best practice guidelines~\cite{Mckinley2011computer}.
Developing computerized disease and treatment models from medical best practice
handbooks needs close interactions with medical professionals.
In addition, to satisfy the safety requirements,
the derived models also need to be both clinically validated and formally verified.
Therefore, how to develop verifiably safe medical guideline models has been
a challenge to both medical professionals and computer scientists/engineers.

Our previous work~\cite{Guo2016ICCPS} proposed and implemented a framework to
support the development of verifiably safe medical guideline models.
The framework seamlessly integrates the interfaces with medical
professionals for clinical validations and computer scientists for formal
verifications, respectively. Fig.~\ref{fig:approach}
depicts the high level architecture of the framework.
To meet the safety requirements of medical guideline models, medical professionals
must be involved in the development loop to perform clinical validation.
The framework models medical best practice guidelines with statecharts
and use statecharts to interact with medical professionals for validating
safety properties. However,
for safety-critical medical guideline models, validation by medical professionals
alone is not adequate for ensuring safety, hence formal verification is
required. However, most statecharts, such as \yakindu~\cite{yakindu}, do
not provide formal verification capability.
The framework hence transforms medical guideline statecharts to timed automata 
so that safety properties can be formally verified.
If a safety property is not satisfied, the framework
also provides the capability to trace the failed paths from timed automata
back to statecharts and assists model developers in debugging and correcting the errors.

\begin{figure}[ht]
	\centering
	\includegraphics[width = 0.8\textwidth]{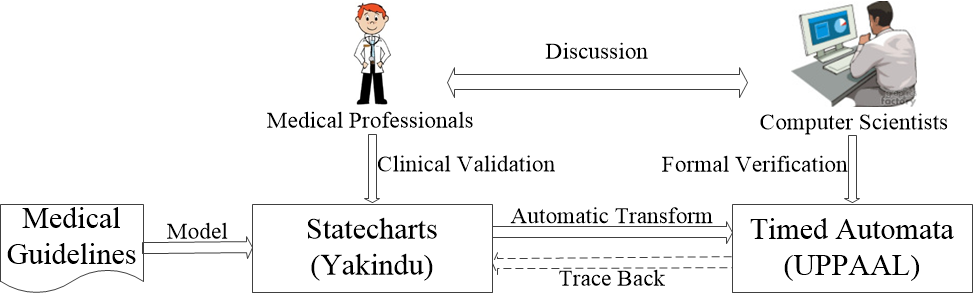}
	\caption{High Level Architecture of Framework}
	\label{fig:approach}
\end{figure}

As shown in Fig.~\ref{fig:approach}, the framework must maintain the model equivalence
between statecharts and transformed timed automata to ensure that
the formal verification results of timed automata also hold for corresponding statecharts.
To prove the model equivalence, all three components in the
framework, i.e., statecharts, timed automata, and the transformation from
statecharts to timed automata, require formal definitions. Behrmann \etal\
formally defined the syntax and semantics of \uppaal\ timed automata in~\cite{behrmann2004tutorial}.
However, the formal semantics of \yakindu\ statecharts and the transformation
in the framework are yet to be defined.

The paper presents the formalism of the framework.
In particular, we first define the statecharts execution
semantics, followed with the formal definitions of statecharts
to timed automata transformation rules. We then 
formally prove that the transformation rules maintain the execution semantic
equivalence between statecharts and transformed timed automata.

The paper is organized as follows.
In Section~\ref{sec:formalism}, we formalize the execution semantics of statecharts.
Section~\ref{sec:y2u} describes the formalized transformation rules
that transform statecharts to timed automata.
In Section~\ref{sec:correctness}, we formally prove the execution semantic
equivalence between statecharts and transformed timed automata.
A case study of a simplified cardiac arrest treatment scenario
is performed in Section~\ref{sec:exp}.
Section~\ref{sec:related} discusses related work.
We conclude in Section~\ref{sec:conclusion}.

\section{Execution Semantics of Basic \yakindu\ Statecharts}
\label{sec:formalism}
The \yakindu\ statecharts provide a set of essential elements
and a set of complex elements. The essential statechart elements
include \textit{states}, \textit{transitions},
\textit{guards}, \textit{actions}, and \textit{variables}.
Other complex elements can be implemented
by model patterns that are built upon the essential elements~\cite{Guo2017CBMS}.
Hence, we focus on the \yakindu\ statecharts only with the
essential elements. We call the statecharts with only the essential
elements \textit{basic \yakindu\ statechart}. We give the formal
definitions in Definition~\ref{def:st} and Definition~\ref{def:stMul}
and use Example~\ref{ex:stEx} to explain the formal definitions.
For readers' convenience, we list
all the notations used in this paper in Table~\ref{tab:TWC}.
For easy understanding, we apply the similar approach
used in \uppaal\ timed automata formalism~\cite{behrmann2004tutorial}
to formalize the syntax and semantics of basic \yakindu\ statecharts.

\begin{table}[ht]
	\caption{Notation Table} 
	\label{tab:TWC}
	\centering  
	\begin{tabular}{ | c | l |} \hline
		\textbf{Notation} & \textbf{Meaning} \\ \hline 
		$\action$/$\actionSet$ & action/action set \\ \hline
		$\step$ & execution index \\ \hline
		$\clock$/$\clockSet$ & clock/clock set \\ \hline
		$\channel$ & synchronization channel \\ \hline
		$\inStateActionAssign$/$\outStateActionAssign$ & state entry/exit action \\ \hline 		
		$\event$ & event \\ \hline
		$\guard$/$\guardSet$ & guard/guard set \\ \hline		
		$\tranPriority$/$\tranPrioritySet$ & transition priority/transition priority set \\ \hline
		$\invariant$ & clock invariant \\ \hline	
		$\statechartPriority$ & statechart priority \\ \hline		
		$\state$/$\stateSet$ & state/state set \\ \hline		 
		$\tran$/$\tranSet$ & transition/transition set \\ \hline	
		$\outTran$/$\outTranSet$ & outgoing transition/outgoing transition set \\ \hline	
		$\highPriTran$/$\highPriTranSet$ & higher priority transition/higher priority transition set \\ \hline		
		$\timer$ & timing trigger \\ \hline		 
		$\trace$ & execution trace \\ \hline
		$\sysstate$/$\sysstateSet$ & system status/system status set \\ \hline		
		$\var$/$\varSet$ & variable/variable set \\ \hline
		$\nu$ & valuation \\ \hline
		$\automata$/$\automataSet$ & timed automata/timed automata set \\ \hline
		$\eventAutomata$/$\eventAutomataSet$ & event automata/event automata set \\ \hline		
		$\timerAutomata$/$\timerAutomataSet$ & timing trigger automata/timing trigger automata set \\ \hline		
		$\statechart$/$\statechartSet$ & statechart/statechart set \\ \hline		 		
	\end{tabular} 
\end{table}

\begin{definition}[Basic \yakindu\ Statechart]
	\label{def:st}	
	A basic \yakindu\ statechart $\statechart$ is a tuple $(\stateSet,
	\state_0, \tranSet, \guardSet, \actionSet, \varSet, \tranPrioritySet, \inStateActionAssign, \outStateActionAssign)$,
	where $\stateSet$, $\guardSet$, $\actionSet$, $\varSet$, and $\tranPrioritySet$
	are a set of states, guards, actions, variables, and
	transition priorities, respectively;
	$\state_0 \in \stateSet$ is the initial state;	
	$\tranSet \subseteq \stateSet \times (\guardSet \times \actionSet \times \tranPrioritySet) \times \stateSet$
	is a set of transitions between states with a guard, an action,
	and a priority of corresponding transition;	
	$\inStateActionAssign \subseteq \stateSet \times \actionSet$
	assigns entry actions to states; and	
	$\outStateActionAssign \subseteq \stateSet \times \actionSet$
	assigns exit actions to states.	
	The transition priority
	$\tranPriority \in \tranPrioritySet$ is represented by a positive integer, the smaller
	the value of $\tranPriority$, the higher the priority of corresponding
	transition $\tran$.
\end{definition}

\begin{definition}[Network of Basic \yakindu\ Statecharts]
	\label{def:stMul}
	A network of basic \yakindu\ statecharts $\statechartSet$ with
	$n$ basic \yakindu\ statecharts is defined as
	$\statechartSet=\{ (\statechart_i, \statechartPriority_i) | 1 \le i \le n \}$,
	where $\statechart_i = (\stateSet_i, \state^0_i, \tranSet_i, \guardSet_i, \actionSet_i, \varSet_i, \tranPrioritySet_i, \inStateActionAssign^i, \outStateActionAssign^i)$ is a basic \yakindu\ statechart, and
	$\statechartPriority_i$ is a positive integer representing the statechart
	$\statechart_i$'s execution priority, the smaller the value of $\statechartPriority_i$,
	the higher the execution priority of $\statechart_i$. Furthermore,
	$\forall (\statechart_i,\statechartPriority_i), (\statechart_j,\statechartPriority_j) \in \statechartSet: \varSet_i = \varSet_j$, i.e., 
	all statecharts in $\statechartSet$ share the same variable set $\varSet$.
\end{definition}

\begin{example}
	\label{ex:stEx}
	We use the \yakindu\ statechart model shown in Fig.~\ref{fig:stEx}
	to explain Definition~\ref{def:st} and Definition~\ref{def:stMul}.
	The model contains two statecharts
	$\statechart_1$ and $\statechart_2$. The statechart $\statechart_1$
	has higher execution priority than the statechart $\statechart_2$,
	hence we assign $\statechartPriority_1$ and $\statechartPriority_2$
	to be 1 and 2, respectively. The model declares two variables:
	an integer variable $x$ and an event variable $\mathtt{eventA}$,
	hence the variable set is $\varSet_{\y} = \{ x, \mathtt{eventA} \}$.
	
	\begin{figure}[ht]
		\centering
		\includegraphics[width = 0.7\textwidth]{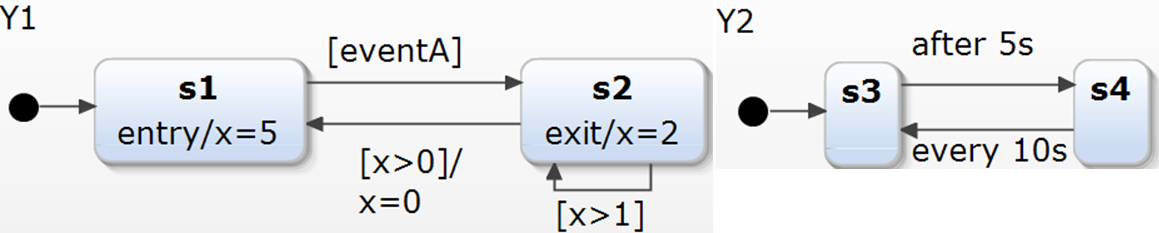}
		\caption{\yakindu\ Statechart Model}
		\label{fig:stEx}
	\end{figure}
	
	The statechart $\statechart_1$ contains three states:
	the initial state $\state_0^1$ (depicted as a filled circle in Fig.~\ref{fig:stEx}),
	state $\state_1$, and state $\state_2$, hence the state set of
	statechart $\statechart_1$ is $\stateSet_{\y}^1 = \{ \state_0^1, \state_1, \state_2 \}$.
	The states $\state_1$ and $\state_2$ have an entry action $x=5$
	and an exit action $x=2$, respectively, hence
	$\inStateActionAssign^1 = \{ (\state_1,x=5) \}$ and
	$\outStateActionAssign^1 = \{ (\state_2,x=2) \}$.
	The statechart $\statechart_1$ contains four transitions. 
	We use $\tran_{\y}^1$ to denote the transition from the initial state
	$\state_0^1$  to state $\state_1$. According to \yakindu\
	statechart's documentation~\cite{yakinduDoc}, the initial state only
	has a outgoing transition which is guarded by $\mathtt{true}$ and
	does not have any actions. We assign the transition priority
	of $\tran_{\y}^1$ as $\tranPriority_1 = 1$.
	Hence, the formal representation of transition $\tran_{\y}^1$ is
	$\tran_{\y}^1 = (\state_0^1, \mathtt{true}, \mathtt{NULL}, 1, \state_1)$.
	Similarly, we formally represent the other three transitions
	of statechart $\statechart_1$ as follows:
	$\tran_{\y}^2 = (\state_1, \mathtt{eventA}, \mathtt{NULL}, 1, \state_2)$,
	$\tran_{\y}^3 = (\state_2, x>0, x=0, 1, \state_1)$, and
	$\tran_{\y}^4 = (\state_2, x>1, \mathtt{NULL}, 2, \state_2)$.
	The state $\state_2$ has two outgoing transitions $\tran_{\y}^3$ and $\tran_{\y}^4$.
	The transition $\tran_{\y}^3$ has higher priority than transition $\tran_{\y}^4$,
	hence we assign the transition priorities as
	$\tranPriority_3 = 1$ and $\tranPriority_4 = 2$.
	The formal representation of statechart $\statechart_1$ is
	$\statechart_1 = (\stateSet_{\y}^1, \state_0^1, \tranSet_{\y}^1, \guardSet_{\y}^1, \actionSet_{\y}^1, \varSet_{\y}, \tranPrioritySet_1, \inStateActionAssign^1, \outStateActionAssign^1)$,
	where $\tranSet_{\y}^1 = \{ \tran_{\y}^1,\tran_{\y}^2,\tran_{\y}^3,\tran_{\y}^4 \}$,
	$\guardSet_{\y}^1 = \{ \mathtt{true}, \mathtt{eventA}, x>0, x>1 \}$,
	$\actionSet_{\y}^1 = \{ x=0, x=5, x=2 \}$, and
	$\tranPrioritySet_1 = \{ 1,2 \}$.
	
	Similarly, we formally represent the statechart $\statechart_2$ as
	$\statechart_2 = (\stateSet_{\y}^2, \state_0^2, \tranSet_{\y}^2, \guardSet_{\y}^2, \actionSet_{\y}^2, \varSet_{\y}, \tranPrioritySet_2, \emptyset, \emptyset)$,
	where
	$\stateSet_{\y}^2 = \{ \state_0^2, \state_3, \state_4 \}$,
	$\tranSet_{\y}^2 = \{ \tran_{\y}^5, \tran_{\y}^6, \tran_{\y}^7 \}$,
	$\tran_{\y}^5 = (\state_0^2, \mathtt{true}, \mathtt{NULL}, 1, \state_3)$,
	$\tran_{\y}^6 = (\state_3, \mathtt{after \ 5s}, \mathtt{NULL}, 1, \state_4)$,
	$\tran_{\y}^7 = (\state_4, \mathtt{every \ 10s}, \mathtt{NULL}, 1, \state_3)$,
	$\guardSet_{\y}^2 = \{ \mathtt{true}, \mathtt{after \ 5s},\mathtt{every \ 10s} \}$,
	$\actionSet_{\y}^2 = \emptyset$, and
	$\tranPrioritySet_2 = \{ 1 \}$.
	Therefore, the formal representation of the \yakindu\
	statechart model shown in Fig.~\ref{fig:stEx} is
	$\statechartSet=\{ (\statechart_1, 1), (\statechart_2, 2) \}$.
\end{example}

Before formalizing the execution semantics of basic \yakindu\ statecharts,
we give the definitions of \textit{valuation function}, \textit{system status}\footnote{The concept of system status is the same as system state. To avoid confusion with statechart state, we call it system status.},
and \textit{labeled transition system}, which are used to define execution semantics,
as follows. 

\begin{definition}[Valuation Function]
	\label{def:valFun}
	Given a set of variables $\varSet$,
	we define $\valueFun{\varSet}$ as a valuation function that maps
	the set of variables $\varSet$ to their corresponding values.
\end{definition}

\begin{definition}[System Status]
	\label{def:sysState}	
	Given a network of basic \yakindu\ statecharts
	$\statechartSet=\{ (\statechart_i, \statechartPriority_i) | \statechart_i = (\stateSet_i, \state^0_i, \tranSet_i, \guardSet_i, \actionSet_i,\\ \varSet, \tranPrioritySet_i, \inStateActionAssign^i, \outStateActionAssign^i) \wedge 1 \le i \le n \}$,
	the system status $\sysstate =(\overline{\state}, \valueFun{\varSet})$
	denotes the active states of each statechart $\statechart_i$
	and values of all variables in $\varSet$,
	where $\overline{\state} = \{ \state_1, \state_2, \dots, \state_n \}$
	is the vector of each statechart $\statechart_i$'s active state
	and $\valueFun{\varSet}$ indicates the values of all variables in $\varSet$.
\end{definition}

\begin{definition}[Labeled Transition System~\cite{Keller1976FVP,Gorrieri2017LTS}]
	\label{def:LTS}
	A labeled transition system (LTS) is a triple
	$\langle \sysstateSet, \sysstate_0, \labelSet, \rightarrow \rangle$, where
	$\sysstateSet$ is a set of system statuses,
	$\sysstate_0$ is the initial system status,
	$\labelSet$ is a set of labels, and
	$\rightarrow \subseteq \sysstateSet \times \labelSet \times \sysstateSet$ is a transition relation.
\end{definition}

Given a labeled transition system $\langle \sysstateSet, \sysstate_0, \labelSet, \rightarrow \rangle$,
we write $\sysstate \xrightarrow{\labels} \sysstate'$ as a short notation for $(\sysstate, \labels, \sysstate') \in \rightarrow$.

The \yakindu\ statechart executions
are implemented by execution cycles. In each execution cycle,
every statechart only executes one step to guarantee synchrony.
Additionally, each statechart is assigned
a unique priority and statecharts executions are sequentialized based on
the assigned priority. Further more, the \yakindu\
statecharts also assign a priority to each transition. Only the
transition with the highest priority from all enabled outgoing
transitions of the same state is selected to perform.
Hence, the \yakindu\ statecharts' execution semantics is synchronous
and deterministic. We define the semantics of
a basic \yakindu\ statechart and a network of basic \yakindu\ statecharts
as follows.

\begin{definition}[Semantics of Basic \yakindu\ Statechart]
	\label{def:stSen1}
	Let $\statechart = (\stateSet, \state_0, \tranSet, \guardSet, \actionSet, \varSet, \tranPrioritySet, \inStateActionAssign, \outStateActionAssign)$
	be a basic \yakindu\ statechart.
	The semantics of statechart $\statechart$ is defined as a labeled transition system $\langle \sysstateSet, \sysstate_0, \{\guardSet \cup \actionSet \cup \tranPrioritySet\}, \rightarrow \rangle$,
	where $\sysstateSet \subseteq \stateSet \times \valueFun{\varSet}$ is a set of system status,
	$\valueFun{\varSet}$ is the valuation function of variables $\varSet$,	
	$\sysstate_0 = (\state_0, \indexValueFun{0}{\varSet})$ is the initial system status,
	$\indexValueFun{0}{\varSet}$ denotes the initial values of all variables in $\varSet$,
	and $\rightarrow \subseteq \sysstateSet \times (\guardSet \times \actionSet \times \tranPrioritySet) \times \sysstateSet$
	is the transition relation defined by
	\begin{align}
	\label{eq:Ysenmantics1}
	\begin{split}
	&(\state, \valueFun{\varSet}) \xrightarrow{\guard,<\outAction;\tranAction;\inAction>,\tranPriority} (\state', \valueFun{\varSet}[<\outAction;\tranAction;\inAction>]) \\
	& \qquad \mathtt{ if } \ \exists (\state, \guard, \tranAction, \tranPriority, \state') \in \tranSet :
	(\valueFun{\varSet} \models \guard) \wedge \prifun(\state, \state', \tranPriority, \valueFun{\varSet}),
	\end{split}	
	\end{align}	
	where $\state, \state' \in \stateSet$;
	$\outAction$, $\tranAction$, and $\inAction$ are the actions of
	exiting state $\state$, the transition $\tran$, and entering
	state $\state'$, respectively;
	the notation $<\outAction;\tranAction;\inAction>$ indicates that the three
	actions $\outAction$, $\tranAction$, and $\inAction$ are executed sequentially and atomically;
	$\valueFun{\varSet}[<\outAction;\tranAction;\inAction>]$ means that the variable values are updated by the action $<\outAction;\tranAction;\inAction>$;
	$\valueFun{\varSet} \models \guard$ denotes that $\valueFun{\varSet}$ satisfies the guard $\guard$;
	and
	$\prifun(\state, \state', \tranPriority, \valueFun{\varSet}) \equiv \forall (\state, \guard', *, \tranPriority', *) \in \tranSet : 
	\valueFun{\varSet} \models \guard' \wedge \tranPriority \le \tranPriority'$.
\end{definition}

In Definition~\ref{def:stSen1}, the transition condition
$\prifun(\state, \state', \tranPriority, \valueFun{\varSet})$ requires
that only the transition with the highest priority
among all enabled outgoing transitions of the same state
is triggered. As each transition in \yakindu\ statecharts
has a unique priority, hence the condition
$\prifun(\state, \state', \tranPriority, \valueFun{\varSet})$
guarantees deterministic execution.

\begin{definition}[Semantics of a Network of Basic \yakindu\ Statecharts]
	\label{def:stSen2}
	Let $\statechartSet=\{ (\statechart_i, \statechartPriority_i) | \statechart_i = (\stateSet_i, \state^0_i, \tranSet_i, \guardSet_i, \actionSet_i, \varSet, \tranPrioritySet_i,\\ \inStateActionAssign^i, \outStateActionAssign^i) \wedge \statechartPriority_i = i \wedge 1 \le i \le n \}$
	be a network of $n$ basic \yakindu\ statecharts
	sorted by statechart priority in decreasing order (i.e.,
	the statechart $\statechart_1$ has the highest priority),	
	and	
	$\overline{\state^0} = \{ \state^0_1, \state^0_2, \dots, \state^0_n \}$ be the initial state vector.
	The semantics of the network of statecharts $\statechartSet$ is defined as a labeled transition system $\langle \sysstateSet, \sysstate_0, \{ \guardSet \cup \actionSet \cup \tranPrioritySet \cup \{ \statechartPriority \}, \rightarrow \rangle$,
	where $\sysstateSet \subseteq (\stateSet_1 \times \stateSet_2 \times \dots \times \stateSet_n) \times \valueFun{\varSet}$ is a set of system status,
	$\valueFun{\varSet}$ is the valuation function of variables $\varSet$,	
	$\sysstate_0 = (\overline{\state^0}, \indexValueFun{0}{\varSet})$	is the initial system status,
	$\indexValueFun{0}{\varSet}$ denotes the initial values of all variables in $\varSet$,
	and $\rightarrow \subseteq \sysstateSet \times (\guardSet \times \actionSet \times \tranPrioritySet \times \{ \statechartPriority \}) \times \sysstateSet$
	is the transition relation defined by	
	\begin{align}
	\label{eq:Ysenmantics2-1}	
	\begin{split}
	&(\overline{\state},\valueFun{\varSet}) \xrightarrow{\statechartPriority_i} (\overline{\state}, \valueFun{\varSet}[\increaseStep]) \\
	&\qquad \mathtt{ if } \
	\statechartPriority_i==\step \ \wedge \
	\forall (\state_i, \guard, *, *, *) \in \tranSet_i : \valueFun{\varSet} \not\models \guard
	\end{split}					
	\end{align}
	and
	\begin{align}
	\label{eq:Ysenmantics2-2}
	\begin{split}
	&(\overline{\state},\valueFun{\varSet})
	\xrightarrow{\guard,<\outAction;\tranAction;\inAction>,\tranPriority,\statechartPriority_i}
	(\overline{\state}[\state'_i/\state_i], \valueFun{\varSet}[<\outAction;\tranAction;\inAction;\increaseStep>])\\
	&\quad \mathtt{if} \ \statechartPriority_i==\step \ \wedge \
	(\exists (\state_i, \guard, \tranAction, \tranPriority, \state'_i) \in \tranSet_i :
	\valueFun{\varSet} \models \guard \wedge \prifun(\state_i, \state'_i, \tranPriority, \valueFun{\varSet})),
	\end{split}			
	\end{align}
	where $\overline{\state}[\state'_i/\state_i]$ denotes that the $i$th element
	$\state_i$ of vector $\overline{\state}$ is replaced by $\state'_i$,
	$\step$ indicates the index of the statechart to be executed, $\indexValueFun{0}{\step} = 1$,
	and $\increaseStep \equiv (\step+1) \mod n$.
\end{definition}

In Definition~\ref{def:stSen2}, the transition condition $\statechartPriority_i=\step$ in
formulas~\eqref{eq:Ysenmantics2-1}and~\eqref{eq:Ysenmantics2-2}
guarantees the deterministic execution semantics, i.e., only the statechart
whose pre-assigned priority is the same with the statechart
execution index $\step$ is executed.
The valuation update function $\increaseStep$ guarantees the synchronous
execution semantics, i.e., every statechart only executes one step
in each \yakindu\ execution cycle.
Hence, the two formulas $\statechartPriority_i=\step$
and $\increaseStep$ together guarantees that each statechart only
executes one step sequentially based on the assigned priority,
i.e., deterministic and synchronous execution.
The formula~\eqref{eq:Ysenmantics2-1} denotes that
there is no transition enabled in the statechart $\statechart_i$,
hence the only action is to update the statechart
execution index $\step$ by $\increaseStep$ to ensure synchronous execution.
The formula~\eqref{eq:Ysenmantics2-2} represents the transition triggered
by guard satisfaction in statechart $\statechart_i$.

\section{Transformations from Statecharts to Timed Automata}
\label{sec:y2u}
Our previous work~\cite{Guo2016ICCPS} presented a set of transformation
rules from statecharts to timed automata.
In this section, we formalize the transformation rules.
In particular, we first introduce
the syntax and execution semantics of \uppaal\ timed automata.
Then, based on the formal representations of \yakindu\ statecharts
and \uppaal\ timed automata, we give the formal definitions
of these transformation rules.

\subsection{Introduction of \uppaal\ Timed Automata}
\label{subsec:y2u-uppaal}
\uppaal\ is a verification toolbox
based on the timed automata theory~\cite{alur1994theory}.
For self-containment, we briefly introduce the syntax and semantics
of \uppaal\ timed automata given in~\cite{behrmann2004tutorial}.
For consistence, we use the same notations, which are used in basic
\yakindu\ statechart definitions in Section~\ref{sec:formalism},
to represent the equivalent elements in \uppaal\ timed automata.

\begin{definition}[Timed Automaton~\cite{behrmann2004tutorial}]
	\label{def:ta}	
	A timed automaton $\automata$ is a tuple $(\stateSet, \state_0, \tranSet, \guardSet, \actionSet, \varSet, \clockSet, \invariant)$,
	where $\stateSet$ is a set of locations, $\state_0 \in \stateSet$ is the initial location,	
	$\actionSet$ is a set of actions, $\guardSet$ is a set of guards, $\varSet$ is a set of variables,
	$\clockSet$ is a set of clocks, $\tranSet \subseteq \stateSet \times (\guardSet \times \actionSet \times \clockSet) \times \stateSet$
	is a set of edges between locations with a guard, an action, and a set of clocks to be reset,
	and $\invariant \subseteq \stateSet \times \clockSet$ assigns clock invariants to locations.
\end{definition}

\begin{definition}[Network of Timed Automata~\cite{behrmann2004tutorial}]
	\label{def:taMul}
	A network of timed automata $\automataSet$ with $n$ timed automata is defined as
	$\automataSet = \{ \automata_i | 1 \le i \le n \}$,
	where $\automata_i = (\stateSet_i, \state^0_i, \tranSet_i, \guardSet_i, \actionSet_i, \varSet_i, \clockSet_i, \invariant_i)$ and	
	$\forall \automata_i, \automata_j \in \automataSet: \varSet_i = \varSet_j \wedge \clockSet_i = \clockSet_j$
	(i.e., all timed automata in $\automataSet$ share the same variable set $\varSet$ and the same clock set $\clockSet$).
\end{definition}

The concepts of \textit{valuation function}, \textit{system status},
and \textit{labeled transition system}, i.e., Definitions~\ref{def:valFun}-\ref{def:LTS},
are also used to define the semantics of \uppaal\ timed automata.

\begin{definition}[Semantics of Timed Automaton~\cite{behrmann2004tutorial}]
	\label{def:taSen1}
	Let $\automata = (\stateSet, \state_0, \tranSet, \guardSet, \actionSet, \varSet, \clockSet, \invariant)$
	be a timed automaton. The semantics is defined as a labeled transition system $\langle \sysstateSet, \sysstate_0, \{ \mathbb{R}_{\ge 0} \cup \actionSet \cup \guardSet \cup \clockSet \}, \rightarrow \rangle$,
	where $\sysstateSet \subseteq \stateSet \times (\valueFun{\varSet} \times \valueFun{\clockSet})$ is a set of system status, $\sysstate_0 = (\state_0, \{ \indexValueFun{0}{\varSet} \cup \indexValueFun{0}{\clockSet} \})$
	is the initial system status, $\indexValueFun{0}{\varSet}$ and $\indexValueFun{0}{\clockSet}$ denote the initial values of all variables in $\varSet$ and all clocks in $\clockSet$, respectively, and $\rightarrow \subseteq \sysstateSet \times (\mathbb{R}_{\ge 0} \times \actionSet \times \guardSet \times \clockSet) \times \sysstateSet$
	is the transition relation defined by
	\begin{align}
	\label{eq:Usenmantics1}
	\begin{split}
	&(\state, \valueFun{\clockSet}) \xrightarrow{d} (\state, \valueFun{\clockSet}+d) \\
	&\qquad \mathtt{ if } \ \forall d' : 0 \le d' \le d \implies
	\valueFun{\clockSet}+d' \models \invariant(\state)
	\end{split}
	\end{align}	
	and
	\begin{align}
	\label{eq:Usenmantics2}
	\begin{split}
	&(\state, \valueFun{\varSet}) \xrightarrow{\guard,\action, \clock} (\state', \valueFun{\varSet}[\action]) \\
	&\qquad \mathtt{ if } \ \exists  (\state, \guard, \action, \clock, \state') \in \tranSet : \valueFun{\varSet} \models \guard \wedge \\
	&\qquad\qquad \valueFun{\clockSet}[\action]= \valueFun{\clockSet}[\clock := 0] \wedge
	\valueFun{\clockSet}[\action] \models \invariant(\state),
	\end{split}	
	\end{align}	
	where $d$ is non-negative real number,
	$\clock := 0$ denotes reset the clock $c$ to be $0$,
	$\valueFun{\clockSet}$ is a clock valuation function from the set of clocks to the non-negative real numbers,
	$\valueFun{\clockSet} \models \invariant(\state)$ indicates $\valueFun{\clockSet}$ satisfies invariant $\invariant(\state)$,
	and $\valueFun{\varSet} \models \guard$ means $\valueFun{\varSet}$ satisfies guard $\guard$.
\end{definition}

In Definition~\ref{def:taSen1}, the transitions represented by formula~\eqref{eq:Usenmantics1}
and formula~\eqref{eq:Usenmantics2} denotes time-passing and location transitions, respectively.

\begin{definition}[Semantics of a Network of Timed Automata~\cite{behrmann2004tutorial}]
	\label{def:taSen2}
	Let $\automataSet = \{ \automata_i | \automata_i = (\stateSet_i, \state^0_i, \tranSet_i, \guardSet_i, \actionSet_i, \varSet_i, \clockSet_i, \invariant_i) \wedge 1 \le i \le n \}$ be a network of $n$ timed automata,
	and let $\overline{\state^0} = \{ \state^0_1, \state^0_2, \dots, \state^0_n \}$ be the initial location vector.
	The semantics of the network is defined as a labeled transition system $\langle \sysstateSet, \sysstate_0, \{ \mathbb{R}_{\ge 0} \cup \actionSet \cup \guardSet \cup \clockSet \}, \rightarrow \rangle$,
	where $\sysstateSet = (\stateSet_1 \times \stateSet_2 \times \dots \times \stateSet_n) \times (\valueFun{\varSet} \times \valueFun{\clockSet})$
	is a set of system status, $\sysstate_0 = (\overline{\state^0}, \{ \indexValueFun{0}{\varSet} \cup \indexValueFun{0}{\clockSet} \})$ is the initial system status,
	$\indexValueFun{0}{\varSet}$ and $\indexValueFun{0}{\clockSet}$ denote the initial values of all variables in $\varSet$ and all clocks in $\clockSet$, respectively, and $\rightarrow \subseteq \sysstateSet \times (\mathbb{R}_{\ge 0} \times \actionSet \times \guardSet \times \clockSet) \times \sysstateSet$
	is the transition relation defined by:
		\begin{align}
		\label{eq:Usenmantics3}
		\begin{split}
		&(\overline{\state}, \valueFun{\clockSet}) \xrightarrow{d} (\overline{\state}, \valueFun{\clockSet} + d) \\
		&\qquad \mathtt{ if } \ \forall d' : 0 \le d' \le d \implies \valueFun{\clockSet}+d' \models \invariant(\overline{\state}),
		\end{split}
		\end{align}
		\begin{align}
		\label{eq:Usenmantics4}
		\begin{split}
		&(\overline{\state}, \valueFun{\varSet}) \xrightarrow{\guard,\action, \clock} (\overline{\state}[\state_i^{'}/\state_i], \valueFun{\varSet}[\action]) \\
		&\qquad\mathtt{ if } \ \exists (\state_i, \guard, \action, \clock, \state'_i) \in \tranSet_i \ :
		\valueFun{\varSet} \models \guard \wedge \\
		&\qquad\qquad \valueFun{\clockSet}[\action]= \valueFun{\clockSet}[\clock := 0] \wedge
		\valueFun{\clockSet}[\action] \models I(\overline{\state}),
		\end{split}			
		\end{align}
		and
		\begin{align}
		\label{eq:Usenmantics5}
		\begin{split}
		&(\overline{\state}, \valueFun{\varSet}) \xrightarrow{\guard,\action, \channel, \clock}
		(\overline{\state}[\state_i^{'}/\state_i, \state_j^{'}/\state_j], \valueFun{\varSet}[\action_i;\action_j])\\
		& \mathtt{ if } \ \exists 
		(\state_i, \guard_i, \action_i, \channel?, \clock_i, \state'_i) \in \tranSet_i \wedge 
		(\state_j, \guard_j, \action_j, \channel!, \clock_j, \state'_j) \in \tranSet_j \ : \\
		&\qquad
		\valueFun{\varSet} \models (\guard_i \wedge \guard_j) \wedge \valueFun{\clockSet}[\action_i;\action_j]=\valueFun{\clockSet}[\clock_i:=0;\clock_j:=0] \wedge \valueFun{\clockSet}[\action_i;\action_j] \models \invariant(\overline{\state}),
		\end{split}			
		\end{align}	
	where $d$ is non-negative real number, 
	$\clock := 0$ denotes reset the clock $c$ to be $0$,
	$\valueFun{\clockSet}$ is a clock valuation function from the set of clocks to the non-negative real numbers, and
	$\overline{\state}[\state_i^{'}/\state_i]$ denotes that the $i$th element $\state_i$ of vector $\overline{\state}$ is replaced by $\state'_i$.
\end{definition}

In Definition~\ref{def:taSen2}, the transitions represented by formula~\eqref{eq:Usenmantics3}, formular~\eqref{eq:Usenmantics4},
and formula~\eqref{eq:Usenmantics5} denote time-passing, transitions triggered by local guard satisfactions, and transitions triggered by synchronizations, respectively.

\subsection{Transformation Rules}
\label{subsec:y2u-rule}
\yakindu\ statecharts and \uppaal\ timed automata have two major differences:
(1) syntactic difference: they have different syntactic element sets, for instance, \textit{events}
and \textit{timing triggers} are elements in \yakindu\ statecharts and they do not have corresponding elements in \uppaal\ automata;
and (2) execution semantics difference: \yakindu\ model has deterministic
and synchronous execution semantics while the execution of \uppaal\ model is non-deterministic
and asynchronous.

As mentioned in Section~\ref{sec:formalism}, \yakindu\ statecharts
contain essential elements and complex elements (also called syntactic sugars).
The essential elements include \textit{states}, \textit{transitions},
\textit{guards}, \textit{actions}, and \textit{variables}.
For elements that have equivalent elements in \uppaal\
timed automata, i.e., \textit{state}, \textit{transition}, and \textit{variable}
elements, they are transformed by one-to-one syntactic mappings
presented in Rule~\ref{rule:initial}.
We transform the \textit{action} and \textit{guard} elements
through Rule~\ref{rule:action} and Rules~\ref{rule:guard}-\ref{rule:timer},
respectively. Note that the \textit{event} and \textit{timing trigger}
elements are implemented as \textit{guards} in \yakindu\ statecharts.

For complex elements, such as \textit{composite states}
and \textit{choice}, we have implemented model patterns to represent them
with essential elements~\cite{Guo2017CBMS}.
By applying the model patterns~\cite{Guo2017CBMS}, \yakindu\ statechart
models with complex elements can be transformed to \uppaal\ timed automata
with essential element transform rules, i.e., Rules~\ref{rule:initial}-\ref{rule:timer}.

To address the execution semantics difference, we enforce the deterministic
and synchronous execution semantics in transformed \uppaal\ timed
automata through Rule~\ref{rule:priority} and Rule~\ref{rule:synchrony}.
The formal definitions of every transformation rules are presented below.
We use the \yakindu\ statechart model shown in Fig.~\ref{fig:stEx} as an example to illustrate
each transformation rule.
Note that the transformation rules are expected to be applied in sequence.
Before given the transformation rules, we introduce notations
$:=$ to denote the assignment and $A(B)$ to denote the $B$ element of $A$.

\subsubsection{Initialization}
\label{subsubsec:initial}
\text{}

\textbf{Rule~\ref{rule:initial}: Initialization.}
Both \yakindu\ statecharts and \uppaal\ timed automata have \textit{state}
elements (called \textit{location} in \uppaal\ timed automata), \textit{transition} elements
(called \textit{edge} in \uppaal\ timed automata), and \textit{variables}.
We initialize a network of \uppaal\ timed automata by one-to-one mapping
\textit{states}, \textit{transitions}, and \textit{variables} in the given
\yakindu\ statecharts to \textit{locations}, \textit{edges}, and
\textit{variables}, respectively.
The formalism of the initialization rule is defined in Rule~\ref{rule:initial}.
We use Example~\ref{ex:initial} to explain the initialization rule.

\begin{myrule}[Initialization Rule]
	\label{rule:initial}
	Given a network of basic \yakindu\ statecharts
	$\statechartSet=\{ (\statechart_i, \statechartPriority_i) | \statechart_i = (\stateSet_i, \state^0_i, \tranSet_i, \guardSet_i, \actionSet_i,\\ \varSet, \tranPrioritySet_i, \inStateActionAssign^i, \outStateActionAssign^i) \wedge 1 \le i \le n \}$,
	for each basic \yakindu\ statechart $\statechart_i \in \statechartSet$,
	create a \uppaal\ timed automaton $\automata_i = (\statechart_i(\stateSet_i), \statechart_i(\state^0_i), \statechart_i(\tranSet_i), \emptyset, \emptyset, \varSet, \emptyset, \emptyset)$
	and
	denote $\automataSetStep{1} = \{ \automata_i | 1 \le i \le n \}$ as the network of \uppaal\ timed automata after initialization.
\end{myrule}

\begin{example}
	\label{ex:initial}
	We initialize a \uppaal\ timed automata model $\automataSetStep{1}$ for the
	statechart model shown in Fig.~\ref{fig:stEx}, i.e., $\statechartSet=\{ (\statechart_1, 1), (\statechart_2, 2) \}$
	specified in Example~\ref{ex:stEx}.
	Fig.~\ref{fig:initialEx} depicts the initialized timed automata model.
	In \uppaal\ GUI, the initial location is depicted by a double circle.
	\begin{figure}[ht]
		\centering
		\includegraphics[width = 0.7\textwidth]{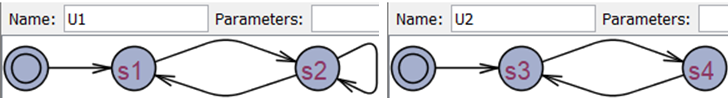}
		\caption{Timed Automata Model after Initialization}
		\label{fig:initialEx}
	\end{figure}
	
	According to Rule~\ref{rule:initial}, for statechart $\statechart_1$,
	we create a timed automaton
	$\automata_1 = (\stateSet_{\u}^1, \state^1_0, \tranSet_{\u}^1, \emptyset, \emptyset, \varSet_{\u}, \emptyset, \emptyset)$
	by one-to-one mapping states, transitions, and variables of $\statechart_1$,
	i.e., $\stateSet_{\u}^1 = \stateSet_{\y}^1 = \{ \state_0^1, \state_1, \state_2 \}$,
	$\tranSet_{\u}^1 = \{ \tran_{\u}^1,\tran_{\u}^2,\tran_{\u}^3,\tran_{\u}^4 \}$,	
	$\tran_{\u}^1 = (\state_0^1, \mathtt{NULL}, \mathtt{NULL}, \mathtt{NULL}, \state_1)$,	
	$\tran_{\u}^2 = (\state_1, \mathtt{NULL}, \mathtt{NULL}, \mathtt{NULL}, \state_2)$,
	$\tran_{\u}^3 = (\state_2, \mathtt{NULL}, \mathtt{NULL}, \mathtt{NULL}, \state_1)$,
	$\tran_{\u}^4 = (\state_2, \mathtt{NULL}, \mathtt{NULL}, \mathtt{NULL},\\ \state_2)$,	
	and $\varSet_{\u} = \varSet_{\y} = \{ x, \mathtt{eventA} \}$.
	
	Similarly, we create a timed automaton
	$\automata_2 = (\stateSet_{\u}^2, \state^0_2, \tranSet_{\u}^2, \emptyset, \emptyset, \varSet_{\u}, \emptyset, \emptyset)$
	for statechart $\statechart_2$, where	
	$\stateSet_{\u}^2 = \{ \state_0^2, \state_3, \state_4 \}$,	
	$\tranSet_{\u}^2 = \{ \tran_{\u}^5, \tran_{\u}^6, \tran_{\u}^7 \}$,	
	$\tran_{\u}^5 = (\state_0^2, \mathtt{NULL}, \mathtt{NULL}, \mathtt{NULL}, \state_3)$,	
	$\tran_{\u}^6 = (\state_3, \mathtt{NULL}, \mathtt{NULL}, \mathtt{NULL}, \state_4)$,	and
	$\tran_{\u}^7 = (\state_4, \mathtt{NULL}, \mathtt{NULL}, \mathtt{NULL}, \state_3)$.
	
	Hence, the timed automata model after initialization is
	$\automataSetStep{1} = \{ \automata_1, \automata_2 \}$.	
\end{example}

\subsubsection{Essential Elements Transformation}
\label{subsubsec:basicRule}
\text{}

\textbf{Rule~\ref{rule:action}: Action Transformation.}
\yakindu\ statecharts support both guard actions and state actions,
while \uppaal\ timed automata only support guard actions.
\yakindu\ statecharts have two types of state actions, \textit{entry action} and
\textit{exit action}, which are carried out on entering or exiting a
state, respectively. The \textit{entry} and \textit{exit} actions of
a state are transformed into actions of all incoming and outgoing
transitions of the state, respectively.
The formalism of the action transformation rule is defined in Rule~\ref{rule:action}.
We use Example~\ref{ex:action} to explain the action transformation rule.

\begin{myrule}[Action Transformation Rule]
	\label{rule:action}
	Given a network of basic \yakindu\ statecharts
	$\statechartSet=\{ (\statechart_i, \statechartPriority_i) | \statechart_i = (\stateSet_i, \state^0_i, \tranSet_i, \guardSet_i, \actionSet_i, \varSet, \tranPrioritySet_i, \inStateActionAssign^i, \outStateActionAssign^i) \wedge 1 \le i \le n \}$
	and the network of \uppaal\ timed automata
	$\automataSetStep{1} = \{ \automata_i | \automata_i = (\stateSet_i, \state^0_i, \tranSet_i, \emptyset, \emptyset, \varSet, \emptyset, \emptyset) \wedge 1 \le i \le n \}$ after initialization,
	let $\tran_{\y}$ denote a transition in statechart $\statechart_i$
	and $\tran_{\u}$ be $\tran_{\y}$'s corresponding edge in the time automata $\automata_i$,
	let $\inState_{\tran_{\y}}$ and $\outState_{\tran_{\y}}$ denote the source state
	and the destination state of transition $\tran_{\y}$, respectively,
	and $\state(\inAction)$ and $\state(\outAction)$ denote the entry actions and exit actions of
	state $\state$ in statechart $\statechart_i$, respectively,
	and let $<\action_1; \action_2>$ denote that the two actions $\action_1$ and $\action_2$ are executed sequentially and atomically,				
	we set each edge $\tran_{\u}$'s action in \uppaal\ timed automata $\automata_i$ with
	an atomic action that consists of 
	exit actions of $\tran_{\y}$'s source state $\inState_{\tran_{\y}}(\outAction)$,
	transition $\tran_{\y}$'s action $\tran_{\y}(\action)$, and entry actions of $\tran_{\y}$'s destination state $\outState_{\tran_{\y}}(\inAction)$ in sequence, i.e.,		
	\begin{align*}
	\begin{split}
	\forall \statechart_i \in \statechartSet: \forall \tran_{\y} \in \statechart_i(\tranSet_i): \tran_{\u}(\action) := <\inState_{\tran_{\y}}(\outAction); \tran_{\y}(\action); \outState_{\tran_{\y}}(\inAction)>	
	\end{split}		
	\end{align*}	
	We use $\automataSetStep{2} = \{ \automata_i | \automata_i = (\stateSet_i, \state^0_i, \tranSet_i, \emptyset, \{\tran_{\u}(\action) | \tran_{\u} \in \tranSet_i \}, \varSet, \emptyset, \emptyset) \wedge 1 \le i \le n \}$ to denote the network of \uppaal\ timed automata after the action transformation.
\end{myrule}

\begin{example}
	\label{ex:action}
	We transform actions in the statechart model $\statechartSet$ specified in Example~\ref{ex:stEx}
	and update the timed automata model $\automataSetStep{1}$ initialized in Example~\ref{ex:initial}.
	Fig.~\ref{fig:actionEx} depicts the timed automata model after action transformation.
	\begin{figure}[ht]
		\centering
		\includegraphics[width = 0.7\textwidth]{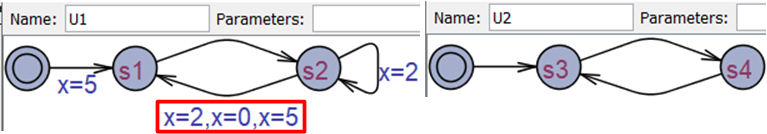}
		\caption{Timed Automata Model after Action Transformation}
		\label{fig:actionEx}
	\end{figure}
	
	We take the transition $\tran_{\y}^3 = (\state_2, x>0, x=0, 1, \state_1)$ of statechart $\statechart_1$
	as an example to illustrate the action transformation rule.
	The action on transition $\tran_{\y}^3$ is $x=0$, i.e., $\tran_{\y}^3(\action) \equiv x=0$.
	The exit action of transition $\tran_{\y}^3$'s source state $\state_2$ is $x=2$, i.e., $\inState_{\tran_{\y}^3}(\outAction) \equiv x=2$.
	The entry action of transition $\tran_{\y}^3$'s destination state $\state_1$ is $x=5$, i.e., $\outState_{\tran_{\y}^3}(\inAction) \equiv x=5$.
	According to Rule~\ref{rule:action}, the action of transition $\tran_{\u}^3$ in timed automaton $\automata_1$
	is $<x=2; x=0; x=5>$, i.e., $\tran_{\u}^3 = (\state_2, \mathtt{NULL}, <x=2; x=0; x=5>, \mathtt{NULL}, \state_1)$, which is marked by a red rectangle in Fig.~\ref{fig:actionEx}.
	
	Similarly, we update transitions $\tran_{\u}^1$ and $\tran_{\u}^4$ as	
	$\tran_{\u}^1 = (\state_0^1, \mathtt{NULL}, x=5, \mathtt{NULL}, \state_1)$ and	
	$\tran_{\u}^4 = (\state_2, \mathtt{NULL}, x=2, \mathtt{NULL}, \state_2)$.
	The action set of timed automaton $\automata_1$ is updated to be $\actionSet_{\u}^1 = \{ <x=2; x=0; x=5>, x=5,x=2 \}$.	
	All other elements of timed automaton $\automata_1$ are not changed.	
	As the statechart $\statechart_2$ does not contain actions,
	the timed automaton $\automata_2$ is the same with Example~\ref{ex:initial}.
	
	Hence, the timed automata model after action transformation is
	$\automataSetStep{2} = \{ \automata_1, \automata_2 \}$, where
	$\automata_1= (\stateSet_{\u}^1, \state^0_1, \tranSet_{\u}^1, \emptyset, \actionSet_{\u}^1, \varSet_{\u},\\ \emptyset, \emptyset)$
	and
	$\automata_2 = (\stateSet_{\u}^2, \state^0_2, \tranSet_{\u}^2, \emptyset, \emptyset, \varSet_{\u}, \emptyset, \emptyset)$.		
\end{example}

\textbf{Rule~\ref{rule:guard}: Guard Transformation.}
The \textit{guard} of a transition in \yakindu\ statecharts
is mapped to the corresponding edge in \uppaal\ timed automata.
If a transition guard contains \textit{events} or \textit{timing triggers},
we transform them with Rule~\ref{rule:event} and Rule~\ref{rule:timer}
given below, respectively.
The formalism of the guard transformation rule is defined in Rule~\ref{rule:guard}.
We use Example~\ref{ex:guard} to explain the guard transformation rule.

\begin{myrule}[Guard Transformation Rule]
	\label{rule:guard}
	Given a network of basic \yakindu\ statecharts
	$\statechartSet=\{ (\statechart_i, \statechartPriority_i) | \statechart_i = (\stateSet_i, \state^0_i, \tranSet_i, \guardSet_i, \actionSet_i, \varSet, \tranPrioritySet_i, \inStateActionAssign^i, \outStateActionAssign^i) \wedge 1 \le i \le n \}$
	and the network of \uppaal\ timed automata
	$\automataSetStep{2} = \{ \automata_i | \automata_i = (\stateSet_i, \state^0_i, \tranSet_i, \emptyset, \actionSet_i, \varSet, \emptyset, \emptyset) \wedge 1 \le i \le n \}$ after \textit{action} transformation,
	let $\tran_{\y}$ denote a transition in statechart $\statechart_i$ and $\tran_{\u}$ be $\tran_{\y}$'s
	corresponding edge in timed automata $\automata_i$,
	we set each edge $\tran_{\u}$'s guard in \uppaal\ timed automata $\automata_i$ with corresponding
	transition $\tran_{\y}$'s guard in \yakindu\ statechart $\statechart_i$, i.e.,				
	\begin{align*}			
	\forall \statechart_i \in \statechartSet: \forall \tran_{\y} \in \statechart_i(\tranSet_i): \tran_{\u}(\guard) := \tran_{\y}(\guard).
	\end{align*}		
	We use $\automataSetStep{3} = \{ \automata_i | \automata_i = (\stateSet_i, \state^0_i, \tranSet_i, \{ \tran_{\u}(\guard) | \tran_{\u}\in \tranSet_i \}, \actionSet_i, \varSet, \emptyset, \emptyset) \wedge 1 \le i \le n \}$ to denote the network of \uppaal\ timed automata after guard transformation.
\end{myrule}

\begin{example}
	\label{ex:guard}
	We transform guards in the statechart model $\statechartSet$ specified in Example~\ref{ex:stEx}
	and update the timed automata model $\automataSetStep{2}$ transformed in Example~\ref{ex:action}.
	Fig.~\ref{fig:guardEx} depicts the timed automata model after guard transformation.
	\begin{figure}[ht]
		\centering
		\includegraphics[width = 0.7\textwidth]{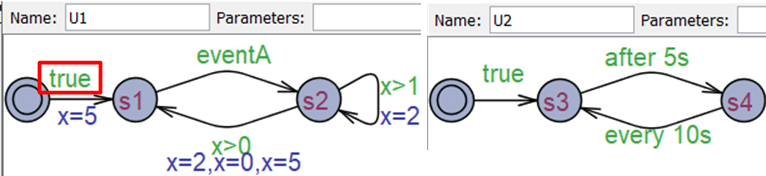}
		\caption{Timed Automata Model after Guard Transformation}
		\label{fig:guardEx}
	\end{figure}
	
	We take the transition $\tran_{\y}^1 = (\state_0^1, \mathtt{true}, \mathtt{NULL}, 1, \state_1)$ as an example to illustrate the guard transformation rule.
	According to Rule~\ref{rule:guard}, the guard of transition $\tran_{\y}^1$'s corresponding edge $\tran_{\u}^1$ is set to be $\tran_{\y}^1$'s guard $\mathtt{true}$, i.e., $\tran_{\u}^1 = (\state_0^1, \mathtt{true}, x=5, \mathtt{NULL}, \state_1)$, which is marked by a red rectangle in Fig.~\ref{fig:guardEx}.
	
	Similarly, we update other six transitions by transforming guards as follows:
	$\tran_{\u}^2 = (\state_1, \mathtt{eventA}, \mathtt{NULL}, \mathtt{NULL}, \state_2)$,
	$\tran_{\u}^3 = (\state_2, x>0, <x=2; x=0; x=5>, \mathtt{NULL}, \state_1)$,
	$\tran_{\u}^4 = (\state_2, x>1, x=2, \mathtt{NULL}, \state_2)$,
	$\tran_{\u}^5 = (\state_0^2, \mathtt{true}, \mathtt{NULL}, \mathtt{NULL}, \state_3)$,	
	$\tran_{\u}^6 = (\state_3, \mathtt{after \ 5s}, \mathtt{NULL}, \mathtt{NULL}, \state_4)$,	and
	$\tran_{\u}^7 = (\state_4, \mathtt{every \ 10s}, \mathtt{NULL}, \mathtt{NULL}, \state_3)$.
	The guard sets of timed automata $\automata_1$ and $\automata_2$ are
	$\guardSet_{\u}^1 = \{ \mathtt{true}, \mathtt{eventA}, x>0, x>1 \}$
	and
	$\guardSet_{\u}^2 = \{ \mathtt{true}, \mathtt{after \ 5s},\mathtt{every \ 10s} \}$, respectively.	
	All other elements of timed automata $\automata_1$ and $\automata_2$ are the same with Example~\ref{ex:action}.	
	
	Hence, the timed automata model after guard transformation is
	$\automataSetStep{3} = \{ \automata_1, \automata_2 \}$, where
	$\automata_1= (\stateSet_{\u}^1, \state^0_1, \tranSet_{\u}^1, \guardSet_{\u}^1, \actionSet_{\u}^1, \varSet_{\u},\\ \emptyset, \emptyset)$
	and
	$\automata_2 = (\stateSet_{\u}^2, \state^0_2, \tranSet_{\u}^2, \guardSet_{\u}^2, \emptyset, \varSet_{\u}, \emptyset, \emptyset)$.		
\end{example}

\textbf{Rule~\ref{rule:event}: Event Transformation.}
In \uppaal\ timed automata, we simulate the event occurrence
by an auxiliary \textit{event} automaton. The event automaton contains
only one location and has a self-loop edge which is synchronized
with the transformed timed automata's edges that are triggered by
the corresponding event. The synchronization between
the transformed timed automata and the event automaton
is through a channel declared for the corresponding event.
The formalism of the event transformation is defined in Rule~\ref{rule:event}.
We use Example~\ref{ex:event} to explain the event transformation rule.

\begin{myrule}[Event Transformation Rule]
	\label{rule:event}
	Given a network of basic \yakindu\ statecharts
	$\statechartSet=\{ (\statechart_i, \statechartPriority_i) | \statechart_i = (\stateSet_i, \state^0_i, \tranSet_i, \guardSet_i, \actionSet_i, \varSet, \tranPrioritySet_i, \inStateActionAssign^i, \outStateActionAssign^i) \wedge 1 \le i \le n \}$
	and the network of \uppaal\ timed automata
	$\automataSetStep{3} = \{ \automata_i | \automata_i = (\stateSet_i, \state^0_i, \tranSet_i, \guardSet_i, \actionSet_i, \varSet, \emptyset, \emptyset) \wedge 1 \le i \le n \}$
	after \textit{guard} transformation,
	let $\varSet_{\event} \subseteq \statechart_i(\varSet)$ denote the event variable set of the network of basic \yakindu\ statecharts $\statechartSet$.
	For each event variable $\event$, we change the variable to be
	channel type, i.e., $\mathtt{chan} \ \event$, and	
	create an \textit{event} automaton
	$\eventAutomata = (\{ \state_0 \}, \state_0,
	\{(\state_0, \event!, \mathtt{NULL}, \mathtt{NULL}, \state_0)\}, \{\event!\}, \emptyset, \varSet, \emptyset, \emptyset)$.
	In addition, for every guard in $\automataSetStep{3}$, we substitute $\event$ with $\event?$, i.e.,
	\begin{align*}
	\forall \automata_i \in \automataSetStep{3}: \forall \guard \in \automata_i(\guardSet_i): \forall \event \in \eventVarSet: \guard[\event? / \event].
	\end{align*}
	We use $\automataSetStep{4} = \{ \automata_i | \automata_i = (\stateSet_i, \state^0_i, \tranSet_i, \guardSet_i[\event?/\event], \actionSet_i, \varSet, \emptyset, \emptyset) \wedge 1 \le i \le n \} \cup \automataSet_{\event}$ to denote the network of \uppaal\ timed automata after event transformation, where $\automataSet_{\event} = \{ \eventAutomata | \event \in \varSet_{\event} \}$ and	
	$\guardSet_i[\event?/\event]$ denotes that each event variable $\event$ of every guard in $\guardSet_i$ is replaced by $\event?$.
\end{myrule}

\begin{example}
	\label{ex:event}
	We transform events in the statechart model $\statechartSet$ specified in Example~\ref{ex:stEx}
	and update the timed automata model $\automataSetStep{3}$ transformed in Example~\ref{ex:guard}.
	Fig.~\ref{fig:eventEx} depicts the timed automata model after event transformation.
	\begin{figure}[ht]
		\centering
		\includegraphics[width = 0.7\textwidth]{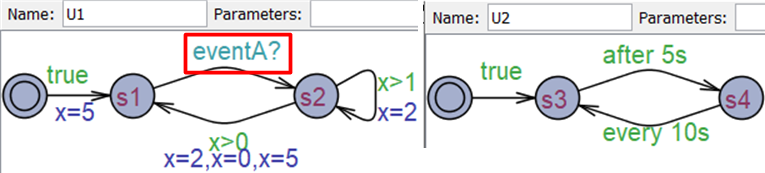}
		\caption{Timed Automata Model after Event Transformation}
		\label{fig:eventEx}
	\end{figure}

	The statechart model $\statechartSet$ only contains one event variable
	$\mathtt{eventA}$. According to Rule~\ref{rule:event}, we first change the
	variable $\mathtt{eventA}$ to be channel type, i.e., re-declare the variable
	by $\mathtt{chan \ eventA}$. Then, we create an auxiliary event automaton
	$\eventAutomata$ with one location and a self-loop edge to
	simulate the occurrence of $\mathtt{eventA}$, as shown in Fig.~\ref{fig:eventAutomataEx}. The formal representation of the
	event automaton is	
	$\eventAutomata = (\stateSet_{\event}, \state_0^{\event}, \tranSet_{\event}, \guardSet_{\event}, \emptyset, \varSet_{\u}, \emptyset, \emptyset)$, where
	$\stateSet_{\event} = \{ \state_0^{\event} \}$,
	$\tranSet_{\event} = \{(\state_0^{\event}, \mathtt{eventA}!, \mathtt{NULL}, \mathtt{NULL}, \state_0^{\event})\}$,
	$\guardSet_{\event} = \{\mathtt{eventA}!\}$, and
	$\varSet_{\u} = \{ x,\mathtt{eventA}\}$.
	\begin{figure}[ht]
		\centering
		\includegraphics[width = 0.1\textwidth]{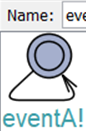}
		\caption{Event Automaton}
		\label{fig:eventAutomataEx}
	\end{figure}

	In the timed automata model $\automataSetStep{3}$, only the transition
	$\tran_{\u}^2 = (\state_1, \mathtt{eventA}, \mathtt{NULL}, \mathtt{NULL}, \state_2)$ of timed automaton $\automata_1$ is triggered by $\mathtt{eventA}$.
	According to Rule~\ref{rule:event}, we update the guard of transition
	$\tran_{\u}^2$ to be $\mathtt{eventA}?$, i.e.,
	$\tran_{\u}^2 = (\state_1, \mathtt{eventA}?, \mathtt{NULL}, \mathtt{NULL}, \state_2)$, which is marked by a red rectangle in Fig.~\ref{fig:eventEx}.
	The guard set of timed automaton $\automata_1$ is updated as
	$\guardSet_{\u}^1 = \{ \mathtt{true}, \mathtt{eventA}?, x>0, x>1 \}$.
	Other elements of $\automata_1$ and all elements of $\automata_2$ are the same with Example~\ref{ex:guard}.
	
	Hence, the timed automata model after event transformation is
	$\automataSetStep{4} = \{ \automata_1, \automata_2, \eventAutomata \}$.	
\end{example}

\textbf{Rule~\ref{rule:timer}: Timing Trigger Transformation.}
\yakindu\ statecharts have two types of timing triggers: \textit{every
timing trigger} and \textit{after timing trigger}.
Similar to the event transformation rule, we simulate
the timing trigger's behavior by an auxiliary timing trigger automaton.
The transformed timed automata and the timing trigger automaton are synchronized
through a channel declared for the corresponding timing trigger.
The formalism of the timing trigger transformation is defined in Rule~\ref{rule:timer}.
We use Example~\ref{ex:timer} to explain the timing trigger
transformation rule.

\begin{myrule}[Timing Trigger Transformation Rule]
	\label{rule:timer}
	Given a network of basic \yakindu\ statecharts
	$\statechartSet=\{ (\statechart_i, \statechartPriority_i) | \statechart_i\\ = (\stateSet_i, \state^0_i, \tranSet_i, \guardSet_i, \actionSet_i, \varSet, \tranPrioritySet_i, \inStateActionAssign^i, \outStateActionAssign^i) \wedge 1 \le i \le n \}$
	and the network of \uppaal\ timed automata
	$\automataSetStep{4} = \{ \automata_i | \automata_i = (\stateSet_i, \state^0_i, \tranSet_i, \guardSet_i, \actionSet_i, \varSet, \emptyset, \emptyset) \wedge 1 \le i \le n \} \cup \automataSet_{\event}$
	after the \textit{event} transformation,
	let $\timer$ be a timing trigger.
	\begin{itemize}			
		\item For each timing trigger $\timer$ in the format of $\mathtt{every} \ \timer \in \{ \statechart_i(\guardSet_i) | \statechart_i \in \statechartSet \}$,
		\begin{itemize}
			\item declare a clock variable $\mathtt{clock} \ \clock_{\timer} = 0$ and a channel variable $\mathtt{chan} \ \var_{\timer}$; and
			\item create an \textit{every} timing trigger automaton
			$\timerAutomata^{\mathtt{every}} = (\{ \state_0 \}, \state_0,
			\{(\state_0, \var_{\timer}!  \ \&\& \    \clock_{\timer} == \timer, \clock_{\timer}=0, \clock_{\timer}, \state_0)\},
			\{\var_{\timer}!  \ \&\& \   \clock_{\timer}\\==\timer\}, \{\clock_{\timer}=0\}, \varSet \cup \{\var_{\timer}\}, \{\clock_{\timer}\}, \{(\state_0, \clock_{\timer} \le \timer)\})$.
		\end{itemize}
		
		\item For each timing trigger $\timer$ in the format of $\mathtt{after} \ \timer \in \{ \statechart_i(\guardSet_i) | \statechart_i \in \statechartSet \}$,
		\begin{itemize}
			\item declare a clock variable $\mathtt{clock} \ \clock_{{\timer}} = 0$ and a channel variable $\mathtt{chan} \ \var_{{\timer}}$; and
			\item create an \textit{after} timing trigger automaton
			$\timerAutomata^{\mathtt{after}} = (\{ \state_0,\state_1 \}, \state_0,
			\{(\state_0, \var_{\timer}! \ \&\& \   \clock_{\timer}==\timer, \clock_{\timer}=0, \clock_{\timer}, \state_1)\},
			\{\var_{\timer}!  \ \&\& \\   \clock_{\timer}==\timer\}, \{\clock_{\timer}=0\}, \varSet \cup \{\var_{\timer}\}, \{\clock_{\timer}\}, \{(\state_0,\clock_{\timer} \le \timer)\})$.
		\end{itemize}
	
		\item For every guard in $\automataSetStep{4}$, we substitute timing triggers
		$\mathtt{every} \ \timer$ and $\mathtt{after} \ \timer$ with corresponding $\var_{\timer}?$, i.e.,
		\begin{align*}
		\begin{split}				
		\forall \automata_i \in \automataSetStep{4}: \forall \guard \in \automata_i(\guardSet_i): \guard[\var_{\timer}?/\mathtt{every} \ \timer][\var_{\timer}?/\mathtt{after} \ \timer].
		\end{split}
		\end{align*}		
	\end{itemize}	
	We use $\automataSetStep{5} = \{ \automata_i | \automata_i = (\stateSet_i, \state^0_i, \tranSet_i, \guardSet_i[\var_{\timer}?/\mathtt{every} \ \timer][\var_{\timer}?/\mathtt{after} \ \timer], \actionSet_i, \varSet \cup \{ \var_{\timer} | \timer \in \statechart_i(\guardSet_i) \wedge \statechart_i \in \statechartSet \} \}, \{ \clock_{\timer} | \timer \in \statechart_i(\guardSet_i) \wedge \statechart_i \in \statechartSet \}, \emptyset) \wedge 1 \le i \le n \} \cup \automataSet_{\event} \cup \automataSet_{\timer}$
	to denote the network of \uppaal\ timed automata after timing trigger transformation,
	where $\automataSet_{\timer} = \{ \timerAutomata | \timer \in \automataSetStep{4}(\guardSet) \}$.
\end{myrule}

\begin{example}
	\label{ex:timer}
	We transform timing triggers in the statechart model $\statechartSet$ specified in Example~\ref{ex:stEx}
	and update the timed automata model $\automataSetStep{4}$ transformed in Example~\ref{ex:event}.
	Fig.~\ref{fig:timerEx} depicts the timed automata model after timing trigger transformation.
	\begin{figure}[ht]
		\centering
		\includegraphics[width = 0.7\textwidth]{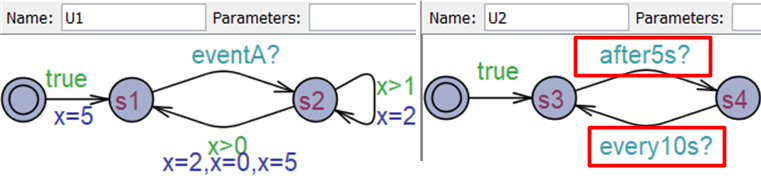}
		\caption{Timed Automata Model after Timing Trigger Transformation}
		\label{fig:timerEx}
	\end{figure}

	The statechart model $\statechartSet$ contains two timing triggers
	$\mathtt{every \ 10s}$ and $\mathtt{after \ 5s}$.
	We take the timing trigger $\mathtt{every \ 10s}$ as an example to illustrate
	Rule~\ref{rule:timer}. We first declare a clock variable $\mathtt{clock} \ \clock_1 = 0$
	and a channel variable $\mathtt{chan \ every10s}$
	for the timing trigger. Then, we add the clock variable $\clock_1$ and the channel variable $\mathtt{every10s}$ into the clock set $\clockSet$
	and the variable set $\varSet_{\u}$ of timed automata model $\automataSetStep{4}$, respectively, i.e.,
	$\clockSet = \{ \clock_1 \}$ and
	$\varSet_{\u} = \{ x, \mathtt{eventA}, \mathtt{every10s} \}$.
	Lastly, we create an every timing trigger automaton
	$\timerAutomata^{\mathtt{every}}$ to simulate the behavior
	of timing trigger $\mathtt{every \ 10s}$, as shown in Fig.~\ref{subfig:everyTimerAutomataEx}.	
	The timing trigger automaton $\timerAutomata^{\mathtt{every}}$ contains only
	one location $\state_0^{\mathtt{every}}$ with invariant $\clock_1 \le 10$
	to guarantee that the timing trigger automaton is not allowed to stay in the state $\state_0^{\mathtt{every}}$ for more than 10 seconds.
	The automaton $\timerAutomata^{\mathtt{every}}$ also has a self-loop
	edge which is guarded by $\clock_1==10$, resets the clock $\clock_1=0$,
	and synchronizes with timed automata in $\automataSetStep{4}$
	through the output channel $\mathtt{every10s} !$.	 
	The formal representation of the timing trigger automaton is
	$\timerAutomata^{\mathtt{every}} = (\{ \state_0^{\mathtt{every}} \}, \state_0^{\mathtt{every}},
	\{(\state_0^{\mathtt{every}}, \mathtt{every10s}! \ \&\& \  \clock_1 == 10, \clock_1=0, \{\clock_1\}, \state_0^{\mathtt{every}})\},
	\{\mathtt{every10s}! \ \&\& \  \clock_1==10\}, \clock_1=0, \varSet_{\u}, \clockSet, \{(\state_0^{\mathtt{every}}, \clock_1 \le 10)\})$.	
	\begin{figure}[ht]
		\centering
		\subfigure[Every Timing Trigger]{
			\centering
			\includegraphics[width = 0.3\textwidth]{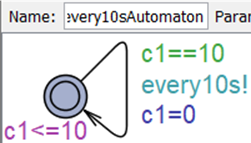}
			\label{subfig:everyTimerAutomataEx}
		}
		\subfigure[After Timing Trigger]{
			\centering
			\includegraphics[width = 0.3\textwidth]{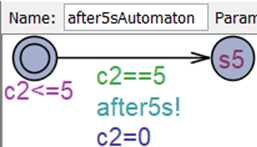}
			\label{subfig:afterTimerAutomataEx}
		}
		\caption{Timing Trigger Automaton}
		\label{fig:timerAutomataEx}
	\end{figure}

	Similarly, we transform the timing trigger $\mathtt{after \ 5s}$.
	We declare a clock variable $\mathtt{clock} \ \clock_2 = 0$
	and a channel variable $\mathtt{chan \ after5s}$
	and add them into the clock set $\clockSet$ and the variable set $\varSet_{\u}$, respectively, i.e.,
	$\clockSet = \{ \clock_1, \clock_2 \}$ and
	$\varSet_{\u} = \{ x, \mathtt{eventA}, \mathtt{every10s}, \mathtt{after5s} \}$. Instead of the self-loop edge in automaton $\timerAutomata^{\mathtt{every}}$,
	the after timing trigger automaton $\timerAutomata^{\mathtt{after}}$
	transits the initial state $\state_0^{\mathtt{after}}$
	to another state $\state_5$, as shown in Fig.~\ref{subfig:afterTimerAutomataEx}.
	The formal representation of the after timing trigger automaton is
	$\timerAutomata^{\mathtt{after}} = (\{ \state_0^{\mathtt{after}},\state_5 \}, \state_0^{\mathtt{after}},
	\{(\state_0^{\mathtt{after}}, \mathtt{after5s}! \ \&\& \  \clock_2==5, \clock_2=0, \clock_2, \state_5)\},
	\{\mathtt{after5s}! \ \&\& \  \clock_2==5\}, \{\clock_2=0\}, \varSet_{\u}, \clockSet, \{(\state_0^{\mathtt{after}},\clock_2 \le 5)\})$.
	
	In the timed automata model $\automataSetStep{4}$, the transitions
	$\tran_{\u}^6 = (\state_3, \mathtt{after \ 5s}, \mathtt{NULL}, \mathtt{NULL}, \state_4)$ and
	$\tran_{\u}^7 = (\state_4, \mathtt{every \ 10s},\\ \mathtt{NULL}, \mathtt{NULL}, \state_3)$
	are triggered by the after and every timing triggers, respectively.
	According to Rule~\ref{rule:timer}, we update the guards of transitions
	$\tran_{\u}^6$ and $\tran_{\u}^7$ to be
	$\mathtt{after5s}?$ and $\mathtt{every10s}?$, respectively, i.e.,
	$\tran_{\u}^6 = (\state_3, \mathtt{after5s}?, \mathtt{NULL}, \mathtt{NULL}, \state_4)$ and
	$\tran_{\u}^7 = (\state_4, \mathtt{every10s}?, \mathtt{NULL}, \mathtt{NULL}, \state_3)$, which are marked by red rectangles in Fig.~\ref{fig:timerEx}.	
	The guard set of timed automaton $\automata_2$ is updated as
	$\guardSet_{\u}^2 = \{ \mathtt{true}, \mathtt{after5s}?,\mathtt{every10s}? \}$.	
	Other elements of $\automata_2$ and all elements of $\automata_1$ are the same with Example~\ref{ex:event}.
	
	Hence, the timed automata model after event transformation is
	$\automataSetStep{5} = \{ \automata_1, \automata_2, \eventAutomata, \timerAutomata^{\mathtt{every}}, \timerAutomata^{\mathtt{after}} \}$.		
\end{example}

\subsubsection{Model Determinism and Synchrony in \uppaal}
\label{subsubsec:synchronous}
\text{}

As mentioned in Section~\ref{sec:formalism},
\yakindu\ statecharts have deterministic and synchronous execution
semantics~\cite{yakinduDoc}, while \uppaal\ timed automata's execution
semantics is non-deterministic and asynchronous~\cite{behrmann2004tutorial}.
To maintain the semantics equivalence
between the \yakindu\ statecharts and the transformed \uppaal\ timed
automata, we need to model determinism and synchrony in \uppaal.

To implement the determinism within a statechart, \yakindu\
assigns an unique priority to each transition and select the
transition with the highest priority from all enabled outgoing
transitions of the same state to execute.
We enforce the transition priorities in \uppaal\ by Rule~\ref{rule:priority}
given below.

\textbf{Rule~\ref{rule:priority}: Transition Priority.}
Assume a state has $n$ outgoing transitions $\{\tran_1, \tran_2, \dots, \tran_n\}$ sorted
in decreasing priority order (i.e., the transition $\tran_1$ has the highest priority)
and the guard of transition $\tran_i$ is
denoted as $\guard_i$. To consider transition priorities, we change the
transition guard of $\tran_i$ to be $\guard_i \ \&\& \  !\guard_1 \ \&\& \  !\guard_2 \
\ \&\& \  \ \dots \ \&\& \  !\guard_{i-1}$ to enforce that higher priority transitions take place
before lower priority transitions.
The formalism of the transition priority transformation is defined in Rule~\ref{rule:priority}.
We use Example~\ref{ex:priority} to explain the transition priority
transformation rule.

\begin{myrule}[Transition Priority Rule]
	\label{rule:priority}	
	Given a network of basic \yakindu\ statecharts
	$\statechartSet=\{ (\statechart_i, \statechartPriority_i) | \statechart_i = (\stateSet_i, \state^0_i, \tranSet_i, \guardSet_i, \actionSet_i, \varSet, \tranPrioritySet_i, \inStateActionAssign^i, \outStateActionAssign^i) \wedge 1 \le i \le n \}$
	and the network of \uppaal\ timed automata
	$\automataSetStep{5} = \{ \automata_i | \automata_i = (\stateSet_i, \state^0_i, \tranSet_i, \guardSet_i, \actionSet_i, \varSet, \clockSet, \emptyset
	) \wedge 1 \le i \le n \} \cup \eventAutomataSet \cup \timerAutomataSet$
	after timing trigger transformation,
	let $\tran_{\y}$ denote a transition in statechart $\statechart_i$ and $\tran_{\u}$ be $\tran_{\y}$'s
	corresponding edge in time automata $\automata_i$,
	let $\outTranSet_{\y}(\state)=\{ \outTran_{\y} | \outTran_{\y} \in \statechart_i(\tranSet_i) \wedge \outTran_{\y}=(\state, *, *, *, *) \}$
	denote the outgoing transitions of state $\state$ in statechart $\statechart_i$,
	and let $\highPriTranSet_{\y}(\state, \tran_{\y})=\{ \outTran_{\y} | \outTran_{\y} \in \outTranSet_{\y}(\state) \wedge \tran_{\y}(\tranPriority) > \outTran_{\y}(\tranPriority) \}$ denote the outgoing transitions of state $\state$ with higher priorities than transition $\tran_{\y}$,
	for each edge $\tran_{\u}$ in the network of \uppaal\ timed automata $\automataSetStep{5}$,
	we modify the guard of $\tran_{\u}$ by adding conjunctions of guard negation of
	higher priority outgoing transitions from $\tran_{\y}$'s source state, i.e.,
	\begin{align*}
	\begin{split}	
	\forall \statechart_i \in \statechartSet: \forall \state \in \statechart_i(\stateSet_i):
	\forall \tran_{\y} \in \outTranSet_{\y}(\state): \tran_{\u}(\guard) := \tran_{\u}(\guard) \bigwedge\limits_{j=1}^k \neg \tran_{\u}^{j}(\guard)	
	\end{split}
	\end{align*}
	where $\tran_{\u}^{j}$ is the corresponding edge of transition 	
	$\highPriTran_{\y} \in \highPriTranSet_{\y}(\state, \tran_{\y})$ and
	$k = | \highPriTranSet_{\y}(\state, \tran_{\y}) |$.
	We use $\automataSetStep{6} = \{ \automata_i | \automata_i = (\stateSet_i, \state^0_i, \tranSet_i, \guardSet_i, \actionSet_i, \varSet, \clockSet, \emptyset) \wedge 1 \le i \le n \} \cup \automataSet_{\event} \cup \automataSet_{\timer}$
	to denote the network of \uppaal\ timed automata after transition priority transformation.
\end{myrule}

\begin{example}
	\label{ex:priority}
	We transform transition priorities in the statechart model $\statechartSet$ specified in Example~\ref{ex:stEx}
	and update the timed automata model $\automataSetStep{5}$ transformed in Example~\ref{ex:timer}.
	Fig.~\ref{fig:priorityEx} depicts the timed automata model after timing trigger transformation.
	\begin{figure}[ht]
		\centering
		\includegraphics[width = 0.7\textwidth]{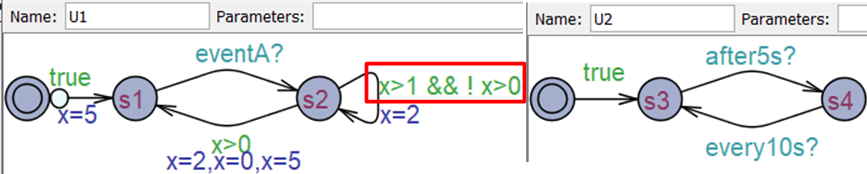}
		\caption{Timed Automata Model after Transition Priority Transformation}
		\label{fig:priorityEx}
	\end{figure}
	
	In the statechart model $\statechartSet$, only the transition
	$\tran_{\y}^4 = (\state_2, x>1, \mathtt{NULL}, 2, \state_2)$
	has a higher priority transition	
	$\tran_{\y}^3 = (\state_2, x>0, x=0, 1, \state_1)$.
	According to Rule~\ref{rule:priority}, we update the guard of transition
	$\tran_{\u}^4$ to be 
	$\tran_{\u}^4(\guard) \wedge \neg \tran_{\u}^3(\guard)$, i.e.,
	$\tran_{\u}^4 = (\state_2, x>1 \ \&\& \  !x>0, x=2, \mathtt{NULL}, \state_2)$,
	which is marked by a red rectangle in Fig.~\ref{fig:priorityEx}.		
	The guard set of timed automaton $\automata_1$ is updated as
	$\guardSet_{\u}^1 = \{ \mathtt{true}, \mathtt{eventA}?, x>0, x>1 \ \&\& \  !x>0 \}$.	
	Other elements of $\automata_1$ and all elements of $\automata_2$ are the same with Example~\ref{ex:timer}.

	Hence, the timed automata model after transition priority transformation is
	$\automataSetStep{6} = \{ \automata_1, \automata_2, \eventAutomata, \timerAutomata^{\mathtt{every}}, \timerAutomata^{\mathtt{after}} \}$.		
\end{example}

\yakindu\ statecharts implement the determinism among different
statecharts by statechart priority and synchronous execution. According
to the decreasing statechart priority, each statechart only executes one
step in each execution cycle if the transition is enabled, otherwise the
statechart stays in the current state. In \uppaal\ timed automata, we
use the lockstep method~\cite{lockstep}, shown in Rule~\ref{rule:synchrony},
to force synchronous execution based on statechart priorities.
We ignore the added event automata (Rule~\ref{rule:event}) and timing trigger
automata  (Rule~\ref{rule:timer}) when modeling synchrony in \uppaal\,
as the added automata does not affect the model's execution behavior.

\textbf{Rule~\ref{rule:synchrony}: Statechart Priority and Synchrony.}
Suppose a \yakindu\ model contains $n$ basic statecharts
$\{ (\statechart_i, \statechartPriority_i) | \statechartPriority_i\\ = i \wedge 1 \le i \le n \}$
that are sorted by statechart priority in decreasing order, i.e.,
the statechart $\statechart_1$ has the highest priority. The corresponding transformed
$n$ time automata are denoted as $\{ \automata_i | 1 \le i \le n \}$.
The lockstep method~\cite{lockstep} to model synchrony is as follows.
We declare a variable $\step$ with initial value $1$ to indicate
the execution index of the transformed timed automata model.
For each edge in timed automata $\automata_i$, we add a
conjunction $\step == \statechartPriority_i$ to its guard to
enforce that only the automaton whose corresponding statechart's execution priority is the same
with the execution index $\step$ is executed.
We also add the execution index update action $\increaseStep \equiv (\step+1) \mod n$
to every action in the transformed timed automata model
to enforce that every automaton only executes one step
in each execution cycle.

If none of outgoing edges is enabled, the
automaton stays in current state, which is equivalent to that the
automaton executes a self-loop step without any actions. To avoid deadlock,
for each location, we add a self-loop edge which is guarded by the negation of all
outgoing edge guards of the location and does not have any actions. If a location does not
have outgoing edges, the added self-loop edge is guarded
by $\mathtt{true}$. The added self-loop edge has the lowest
priority among all outgoing edges of the corresponding location.
We also apply the above lockstep procedure to these added self-loop edges.
The formalism of the statechart priority and synchrony transformation is defined in Rule~\ref{rule:synchrony}.
We use Example~\ref{ex:synchrony} to explain the statechart priority and synchrony
transformation rule.

\begin{myrule}[Statechart Priority and Synchrony Rule]
	\label{rule:synchrony}
	Given a network of basic \yakindu\ statecharts
	$\statechartSet=\{ (\statechart_i, \statechartPriority_i) | \statechart_i = (\stateSet_i, \state^0_i, \tranSet_i, \guardSet_i, \actionSet_i, \varSet, \tranPrioritySet_i, \inStateActionAssign^i, \outStateActionAssign^i) \wedge \statechartPriority_i = i \wedge 1 \le i \le n \}$
	and the network of \uppaal\ timed automata
	$\automataSetStep{6} = \{ \automata_i | \automata_i = (\stateSet_i, \state^0_i, \tranSet_i, \guardSet_i, \actionSet_i, \varSet, \clockSet, \emptyset) \wedge 1 \le i \le n \} \cup \eventAutomataSet \cup \timerAutomataSet$
	after transition priority transformation,	
	let $\outTranSet(\state)=\{ \outTran | \outTran \in \automata_i(\tranSet_i) \wedge \outTran=(\state, *, *, *, *) \}$
	denote the outgoing edges of location $\state$ in timed automata $\automata_i$,	
	let $<\action_1; \action_2>$ denote that the two actions $\action_1$ and $\action_2$ are executed sequentially and atomically,
	we
	\begin{itemize}
		\item declare an integer $\mathtt{int} \ \step = 1$ to indicate the index of the automaton to be executed;
		\item define an action $\increaseStep \equiv (\step+1) \mod n$ to update the automaton execution index $\step$;
		\item for each location in the network of \uppaal\ timed automata $\automataSetStep{6}$,
		add a self-loop edge which is guarded by the negation of all outgoing edge guards of
		the corresponding location and does not have any actions, i.e.,
		$\forall \automata_i \in \automataSetStep{6} :\ \forall \state \in \automata_i(\stateSet_i):\ \automata_i(\tranSet_i) := \automata_i(\tranSet_i) \cup \tran=(\state, \guard, \mathtt{NULL}, \mathtt{NULL}, \state)$,		
		where $\forall \outTran \in \outTranSet (\state): \guard := \mathtt{true} \bigwedge \neg \outTran(\guard)$;	
		\item add a conjunction $\step == \statechartPriority_i$ to every edge's guard in the network of \uppaal\ timed automata $\automataSetStep{6}$, i.e.,		
		$\forall \automata_i \in \automataSetStep{6} :\
		\forall \tran \in \automata_i(\tranSet_i) :\
		\tran(\guard) := \tran(\guard) \ \&\& \  (\step == \statechartPriority_i)$; and	
		\item add the automaton execution index update action $\increaseStep$ to every edge's action in the network of \uppaal\ timed automata $\automataSetStep{6}$, i.e.,
		$\forall \automata_i \in \automataSetStep{6} :\ 
		\forall \tran \in \automata_i(\tranSet_i) :\
		\tran(\action) := <\tran(\action); \increaseStep>$.
	\end{itemize}
	We use $\automataSetStep{7} = \{ \automata_i | \automata_i = (\stateSet_i, \state^0_i, \tranSet_i, \guardSet_i, \actionSet_i, \varSet, \clockSet, \emptyset) \wedge 1 \le i \le n \} \cup \eventAutomataSet \cup \timerAutomataSet$
	to denote the network of \uppaal\ timed automata after statechart priority and synchrony transformation.	
\end{myrule}

\begin{example}
	\label{ex:synchrony}
	We transform statechart priority and model synchrony in the statechart model $\statechartSet$ specified in Example~\ref{ex:stEx}
	and update the timed automata model $\automataSetStep{6}$ transformed in Example~\ref{ex:priority}.
	Fig.~\ref{fig:synchronyEx} depicts the timed automata model after statechart priority and synchrony transformation.
	\begin{figure}[ht]
		\centering
		\includegraphics[width = 0.9\textwidth]{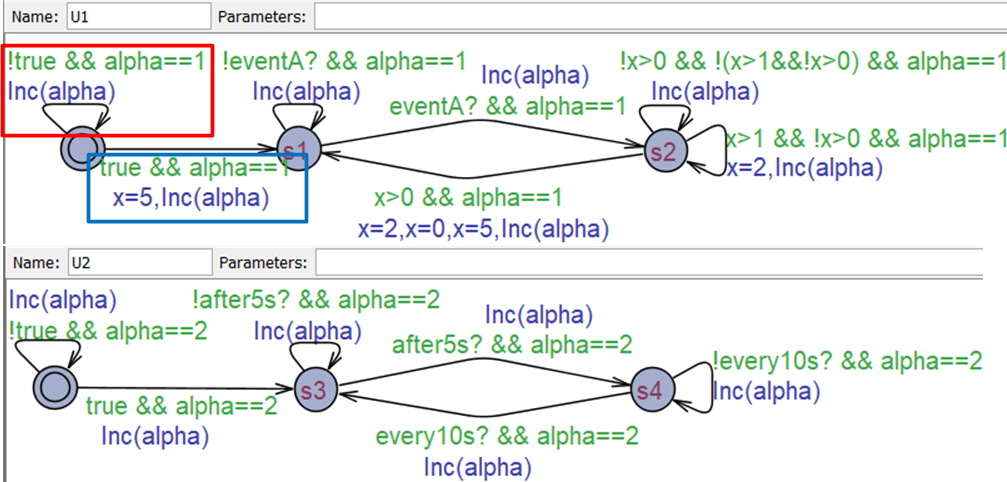}
		\caption{Timed Automata Model after Statechart Priority and Synchrony Transformation}
		\label{fig:synchronyEx}
	\end{figure}

	According to Rule~\ref{rule:synchrony}, we declare an integer $\mathtt{int \ alpha} = 1$ to indicate the automaton execution index	
	and define an action $\mathtt{Inc(alpha)}$ to update $\mathtt{appha}$.
	To avoid deadlock, we add a self-loop edge for each location. We take the
	added self-loop edge $\tran_{\u}^8$ for location $\state_0^1$, which is marked by a red rectangle in Fig.~\ref{fig:synchronyEx},
	as an example to illustrate the rule.
	Before adding edge $\tran_{\u}^8$, the location $\state_0^1$ only has
	one outgoing edge
	$\tran_{\u}^1 = (\state_0^1, \mathtt{true}, x=5, \mathtt{NULL}, \state_1)$.
	We set the guard of $\tran_{\u}^8$ to be $\neg \tran_{\u}^1(\guard)$, i.e.,
	$\tran_{\u}^8 = (\state_0^1, !\mathtt{true}, \mathtt{NULL}, \mathtt{NULL}, \state_0^1)$.
	Similarly, we add self-loop edges  
	$\tran_{\u}^9 = (\state_1, !\mathtt{eventA}?, \mathtt{NULL}, \mathtt{NULL}, \state_1)$,	
	$\tran_{\u}^{10} = (\state_2, !x>0 \ \&\& \  !(x>1 \ \&\& \  !x>0), \mathtt{NULL}, \mathtt{NULL}, \state_2)$,	
	$\tran_{\u}^{11} = (\state_0^2, !\mathtt{true}, \mathtt{NULL}, \mathtt{NULL}, \state_0^2)$,	
	$\tran_{\u}^{12} = (\state_3, !\mathtt{after5s}?, \mathtt{NULL}, \mathtt{NULL}, \state_3)$, and 	
	$\tran_{\u}^{13} = (\state_4, !\mathtt{every10s}?, \mathtt{NULL}, \mathtt{NULL}, \state_4)$	
	for locations $\state_1$, $\state_2$, $\state_0^2$, $\state_3$, and $\state_4$, respectively.
	The edge sets of timed automata $\automata_1$ and $\automata_2$ are updated as
	$\tranSet_{\u}^1 = \{ \tran_{\u}^1,\tran_{\u}^2,\tran_{\u}^3,\tran_{\u}^4,\tran_{\u}^8,\tran_{\u}^9,\tran_{\u}^{10} \}$ and
	$\tranSet_{\u}^2 = \{ \tran_{\u}^5,\tran_{\u}^6,\tran_{\u}^7,\tran_{\u}^{11},\tran_{\u}^{12},\tran_{\u}^{13} \}$, respectively.
	The guard sets of timed automata $\automata_1$ and $\automata_2$ are updated as
	$\guardSet_{\u}^1 = \{ \mathtt{true}, \mathtt{eventA}?, x>0, x>1 \ \&\& \  !x>0, !\mathtt{true}, !\mathtt{eventA}?, !x>0 \ \&\& \  !(x>1 \ \&\& \  !x>0) \}$ and
	$\guardSet_{\u}^2 = \{ \mathtt{true}, \mathtt{after5s}?,\mathtt{every10s}?, !\mathtt{true}, !\mathtt{after5s}?, !\mathtt{every10s}? \}$, respectively.
	
	To model synchrony by applying the lockstep method, we add a conjunction $\mathtt{alpha} == \statechartPriority_i$ and
	the action $\mathtt{Inc(alpha)}$ to every edge's guard and action in timed automata
	model $\automataSetStep{6}$, respectively.
	We take the edge $\tran_{\u}^1 = (\state_0^1, \mathtt{true}, x=5, \mathtt{NULL}, \state_1)$,
	which is marked by a blue rectangle in Fig.~\ref{fig:synchronyEx},
	as an example to illustrate the rule.
	The edge $\tran_{\u}^1$'s corresponding transition in the given \yakindu\ model is transition $\tran_{\y}^1$ of
	statechart $\statechart_1$. According to Rule~\ref{rule:synchrony}, we update
	the guard and action of edge $\tran_{\u}^1$ to be $\tran_{\u}^1(\guard) \ \&\& \  \mathtt{alpha} == \statechartPriority_1$
	and $<\tran_{\u}^1(\action); \mathtt{Inc(alpha)}>$, respectively, i.e.,
	$\tran_{\u}^1 = (\state_0^1, \mathtt{true} \ \&\& \  \mathtt{alpha} == 1, <x=5;\mathtt{Inc(alpha)}>, \mathtt{NULL}, \state_1)$.
	Similarly, we update the other twelve edges in timed automata model
	$\automata_1$ and $\automata_2$ as follows:
	$\tran_{\u}^2 = (\state_1, \mathtt{eventA}? \ \&\& \  \mathtt{alpha} == 1, \mathtt{NULL}, \mathtt{NULL}, \state_2)$,	
	$\tran_{\u}^3 = (\state_2, x>0 \ \&\& \  \mathtt{alpha} == 1, <x=2; x=0; x=5; \mathtt{Inc(alpha)}>, \mathtt{NULL}, \state_1)$,	
	$\tran_{\u}^4 = (\state_2, x>1 \ \&\& \  !x>0 \ \&\& \  \mathtt{alpha} == 1, <x=2; \mathtt{Inc(alpha)}>, \mathtt{NULL}, \state_2)$,
	$\tran_{\u}^5 = (\state_0^2, \mathtt{true} \ \&\& \  \mathtt{alpha} == 2, \mathtt{Inc(alpha)}, \mathtt{NULL}, \state_3)$,
	$\tran_{\u}^6 = (\state_3, \mathtt{after5s}? \ \&\& \  \mathtt{alpha} == 2, \mathtt{Inc(alpha)}, \mathtt{NULL}, \state_4)$,	
	$\tran_{\u}^7 = (\state_4, \mathtt{every10s}? \ \&\& \  \mathtt{alpha} == 2, \mathtt{Inc(alpha)}, \mathtt{NULL}, \state_3)$,	
	$\tran_{\u}^8 = (\state_0^1, !\mathtt{true} \ \&\& \  \mathtt{alpha} == 1, \mathtt{Inc(alpha)}, \mathtt{NULL}, \state_0^1)$,	
	$\tran_{\u}^9 = (\state_1, !\mathtt{eventA}? \ \&\& \  \mathtt{alpha} == 1, \mathtt{Inc(alpha)}, \mathtt{NULL}, \state_1)$,		
	$\tran_{\u}^{10} = (\state_2, !x>0 \ \&\& \  !(x>1 \ \&\& \  !x>0) \ \&\& \  \mathtt{alpha} == 1, \mathtt{Inc(alpha)}, \mathtt{NULL}, \state_2)$,	
	$\tran_{\u}^{11} = (\state_0^2, !\mathtt{true} \ \&\& \  \mathtt{alpha} == 2, \mathtt{Inc(alpha)}, \mathtt{NULL}, \state_0^2)$,	
	$\tran_{\u}^{12} = (\state_3, !\mathtt{after5s}? \ \&\& \  \mathtt{alpha} == 2, \mathtt{Inc(alpha)}, \mathtt{NULL}, \state_3)$, and 	
	$\tran_{\u}^{13} = (\state_4, !\mathtt{every10s}? \ \&\& \  \mathtt{alpha} == 2, \mathtt{Inc(alpha)}, \mathtt{NULL},\\ \state_4)$.	
	The guard sets of timed automata $\automata_1$ and $\automata_2$ are updated as
	$\guardSet_{\u}^1 = \{ \mathtt{true}\ \&\& \  \mathtt{alpha} == 1, \mathtt{eventA}?\ \&\& \  \mathtt{alpha} == 1, x>0, x>1 \ \&\& \  !x>0\ \&\& \  \mathtt{alpha} == 1, !\mathtt{true}\ \&\& \  \mathtt{alpha} == 1, !\mathtt{eventA}?\ \&\& \  \mathtt{alpha} == 1, !x>0 \ \&\& \  !(x>1 \ \&\& \  !x>0)\ \&\& \  \mathtt{alpha} == 1 \}$ and 
	$\guardSet_{\u}^2 = \{ \mathtt{true}\ \&\& \  \mathtt{alpha} == 2, \mathtt{after5s}?\ \&\& \  \mathtt{alpha} == 2,\mathtt{every10s}?\ \&\& \  \mathtt{alpha} == 2, !\mathtt{true}\ \&\& \  \mathtt{alpha} == 2, !\mathtt{after5s}?\ \&\& \  \mathtt{alpha} == 2, !\mathtt{every10s}?\ \&\& \  \mathtt{alpha} == 2 \}$, respectively.
	The action sets of timed automata $\automata_1$ and $\automata_2$ are updated as
	$\actionSet_{\u}^1 = \{ <x=2; x=0; x=5;\mathtt{Inc(alpha)}>, <x=5;\mathtt{Inc(alpha)}>, <x=2; \mathtt{Inc(alpha)}>, \mathtt{Inc(alpha)} \}$ and
	$\actionSet_{\u}^2 = \{ \mathtt{Inc(alpha)} \}$.
	Other elements of timed automata $\automata_1$ and $\automata_2$ are the same with Example~\ref{ex:priority}.

	Hence, the timed automata model after statechart priority and synchrony transformation is
	$\automataSetStep{7} = \{ \automata_1, \automata_2, \eventAutomata,\\ \timerAutomata^{\mathtt{every}}, \timerAutomata^{\mathtt{after}} \}$.		
\end{example}

\section{Execution Semantic Equivalence between Statecharts and Transformed Timed Automata}
\label{sec:correctness}
In this section, we prove that the defined transformation
rules in Section~\ref{subsec:y2u-rule} maintain execution
semantics equivalences between statecharts and
transformed timed automata. 
To do so, we first give the model execution semantics equivalence
definitions. Based on the definition, we then formally prove that
the transformation rules maintain the execution semantics equivalence.

\subsection{Definition of Execution Semantics Equivalence}
\label{subsec:equivalence}
The execution semantics equivalence means that
two models have the same execution behaviors under the same environment.
The formal definitions of \textit{execution trace}, \textit{execution trace equivalence}, and
\textit{model equivalence} are given below.

\begin{definition}[Execution Trace]
	\label{def:trace}
	Given a network of basic \yakindu\ statecharts
	$\statechartSet=\{ (\statechart_i, \statechartPriority_i) | \statechart_i = (\stateSet_i, \state^0_i, \tranSet_i, \guardSet_i, \actionSet_i,\\ \varSet, \tranPrioritySet_i, \inStateActionAssign^i, \outStateActionAssign^i) \wedge 1 \le i \le n \}$
	(or a network of \uppaal\ timed automata
	$\automataSet = \{ \automata_i | \automata_i = (\stateSet_i, \state^0_i, \tranSet_i, \guardSet_i, \actionSet_i, \varSet, \clockSet, \invariant_i) \wedge 1 \le i \le n \}$),
	the execution trace $\trace$ is defined as a sequence of consecutive system statuses
	$\sysstate_0 \rightarrow \sysstate_1 \rightarrow \dots \rightarrow \sysstate_j \rightarrow \dots$,
	where $\sysstate =(\overline{\state}, \valueFun{\varSet})$ denotes the system status.
\end{definition}

\begin{definition}[Execution Trace Equivalence]
	\label{def:traceEq}
	Given two execution traces $\trace_1 = \sysstate_0^1 \rightarrow \sysstate_1^1 \rightarrow \dots \rightarrow \sysstate_i^1 \rightarrow \dots \rightarrow \sysstate_n^1$ and $\trace_2 = \sysstate_0^2 \rightarrow \sysstate_1^2 \rightarrow \dots \rightarrow \sysstate_j^2 \rightarrow \dots \rightarrow \sysstate_m^2$,
	the two execution traces $\trace_1$ and $\trace_2$ are equivalent, denoted as $\trace_1 = \trace_2$,
	if $n=m \wedge (\forall i,j : i=j \implies \sysstate_i^1 = \sysstate_j^2)$,
	where $\sysstate_1(\overline{\state}_1, \indexValueFun{1}{\varSet}) = \sysstate_2(\overline{\state}_2, \indexValueFun{2}{\varSet}) \equiv \overline{\state}_1=\overline{\state}_2 \wedge \indexValueFun{1}{\varSet}=\indexValueFun{2}{\varSet}$.
\end{definition}

\begin{definition}[Model Equivalence]
	\label{def:equivalence}
	Given a network of basic \yakindu\ statecharts
	$\statechartSet=\{ (\statechart_i, \statechartPriority_i) | \statechart_i = (\stateSet_i, \state^0_i, \tranSet_i, \guardSet_i,\\ \actionSet_i, \varSet, \tranPrioritySet_i, \inStateActionAssign^i, \outStateActionAssign^i) \wedge 1 \le i \le n \}$
	and a network of \uppaal\ timed automata
	$\automataSet = \{ \automata_i | \automata_i = (\stateSet_i, \state^0_i, \tranSet_i, \guardSet_i, \actionSet_i, \varSet, \clockSet, \invariant_i)\\ \wedge 1 \le i \le n \}$,
	the two models $\statechartSet$ and $\automataSet$
	are equivalent if they have the equivalent execution trace under the same initial system status, i.e.,
	$\trace_{\statechart} = \trace_{\automata}$ if $\sysstate_{\statechart}^0 = \sysstate_{\automata}^0$,
	where $\trace_{\statechart}$ and $\trace_{\automata}$ are the execution traces of the given \yakindu\ model
	and \uppaal\ model, respectively, and $\sysstate_{\statechart}^0$ and $\sysstate_{\automata}^0$ are the initial
	system status of the given \yakindu\ model and \uppaal\ model, respectively.
\end{definition}

\subsection{Proof of Execution Semantics Equivalence}
\label{subsec:proof}
Our proof strategy is:
(1) prove that the transformation rules maintain one-to-one mappings
of statecharts and states (Lemma~\ref{lm:stMap});
(2) prove that the transformation rules maintain injective non-surjective
mappings of transitions and variables (Lemma~\ref{lm:tranMap});
(3) prove that transformation Rules~\ref{rule:action}-\ref{rule:synchrony}
maintain model equivalence by proving the execution semantics of
\yakindu\ statecharts and \uppaal\ timed automata transformed
by each rule are equivalent (Lemma~\ref{lm:action} to Lemma~\ref{lm:synchrony});
(4) prove that \yakindu\ models with only one statechart are equivalent with
the corresponding transformed \uppaal\ models in Theorem~\ref{thm:singlest}
using induction approach;
and (5) extend Theorem~\ref{thm:singlest} to \yakindu\ models with multiple
statecharts in Theorem~\ref{thm:mulst}.

For the seven transformation rules given in Section~\ref{subsec:y2u-rule},
Rule~\ref{rule:event} and Rule~\ref{rule:timer} add
auxiliary event automata and timing trigger automata to the \uppaal\ model
to simulate the event occurrence and timing trigger, respectively. 
According to the model equivalence definition, 
the auxiliary automata do not affect the model's execution behaviors.
Unless otherwise stated, in the rest of this section when we use the term
\uppaal\ model, we will ignore the auxiliary event automata and
timing trigger automata.

\begin{lemma}
	\label{lm:stMap}
	Given a \yakindu\ model
	$\statechartSet=\{ (\statechart_i, \statechartPriority_i) | \statechart_i = (\stateSet_i, \state^0_i, \tranSet_i, \guardSet_i, \actionSet_i, \varSet, \tranPrioritySet_i, \inStateActionAssign^i, \outStateActionAssign^i) \wedge 1 \le i \le n \}$
	and its transformed \uppaal\ model
	$\automataSet = \{ \automata_i | \automata_i = (\stateSet_i, \state^0_i, \tranSet_i, \guardSet_i, \actionSet_i, \varSet, \clockSet, \invariant_i) \wedge 1 \le i \le n \}$,
	the mappings of statecharts and states from the \yakindu\ model to the \uppaal\ model are bijective, i.e.,	
	the mappings from $\{ \statechart_i | \statechart_i \in \statechartSet \}$ and $\{ \stateSet_i | \stateSet_i \in \statechart_i \wedge \statechart_i \in \statechartSet \}$
	to $\{ \automata_i | \automata_i \in \automataSet \}$ and $\{ \stateSet_i | \stateSet_i \in \automata_i \wedge \automata_i \in \automataSet \}$
	are bijective, respectively.
\end{lemma}

\begin{proof}
	According to Rule~\ref{rule:initial}, every statechart and state
	in the given \yakindu\ model are transformed to an unique
	automata and location in the transformed \uppaal\ model, respectively.
	Hence, the mappings from $\{ \statechart_i | \statechart_i \in \statechartSet \}$ and $\{ \stateSet_i | \stateSet_i \in \statechart_i \wedge \statechart_i \in \statechartSet \}$
	to $\{ \automata_i | \automata_i \in \automataSet \}$ and $\{ \stateSet_i | \stateSet_i \in \automata_i \wedge \automata_i \in \automataSet \}$
	are injective, respectively.
	
	Based on all transformation rules (Rules~\ref{rule:initial}-\ref{rule:synchrony}),
	no additional automata or locations are added into the transformed
	\uppaal\ model. Hence, the mappings from $\{ \statechart_i | \statechart_i \in \statechartSet \}$ and $\{ \stateSet_i | \stateSet_i \in \statechart_i \wedge \statechart_i \in \statechartSet \}$
	to $\{ \automata_i | \automata_i \in \automataSet \}$ and $\{ \stateSet_i | \stateSet_i \in \automata_i \wedge \automata_i \in \automataSet \}$
	are also surjective, respectively.
	
	Therefore, the mappings from $\{ \statechart_i | \statechart_i \in \statechartSet \}$ and $\{ \stateSet_i | \stateSet_i \in \statechart_i \wedge \statechart_i \in \statechartSet \}$
	to $\{ \automata_i | \automata_i \in \automataSet \}$ and $\{ \stateSet_i | \stateSet_i \in \automata_i \wedge \automata_i \in \automataSet \}$
	are bijective, respectively.
\end{proof}

\begin{lemma}
	\label{lm:tranMap}
	Given a \yakindu\ model
	$\statechartSet=\{ (\statechart_i, \statechartPriority_i) | \statechart_i = (\stateSet_i, \state^0_i, \tranSet_i, \guardSet_i, \actionSet_i, \varSet, \tranPrioritySet_i, \inStateActionAssign^i, \outStateActionAssign^i) \wedge 1 \le i \le n \}$
	and its transformed \uppaal\ model
	$\automataSet = \{ \automata_i | \automata_i = (\stateSet_i, \state^0_i, \tranSet_i, \guardSet_i, \actionSet_i, \varSet, \clockSet, \invariant_i) \wedge 1 \le i \le n \}$,
	the mappings of transitions and variables from the \yakindu\ model
	to the \uppaal\ model is injective but not surjective, respectively, i.e.,	
	the mappings from $\{ \tranSet_i | \tranSet_i \in \statechart_i \wedge \statechart_i \in \statechartSet \}$
	and $\{ \varSet | \varSet \in \statechart_i \wedge \statechart_i \in \statechartSet \}$
	to $\{ \tranSet_i | \tranSet_i \in \automata_i \wedge \automata_i \in \automataSet \}$
	and $\{ \varSet | \varSet \in \automata_i \wedge \automata_i \in \automataSet \}$
	are injective but not surjective, respectively.
\end{lemma}

\begin{proof}
	According to Rule~\ref{rule:initial}, every transition and variable
	in the given \yakindu\ model is transformed to an unique
	edge and variable in the transformed \uppaal\ model, respectively.
	Hence, the mappings from
	$\{ \tranSet_i | \tranSet_i \in \statechart_i \wedge \statechart_i \in \statechartSet \}$
	and $\{ \varSet | \varSet \in \statechart_i \wedge \statechart_i \in \statechartSet \}$
	to $\{ \tranSet_i | \tranSet_i \in \automata_i \wedge \automata_i \in \automataSet \}$
	and $\{ \varSet | \varSet \in \automata_i \wedge \automata_i \in \automataSet \}$
	are injective, respectively.
	
	The transformation Rule~\ref{rule:synchrony} adds additional edges
	into the transformed \uppaal\ model. The transformation
	Rule~\ref{rule:timer} and Rule~\ref{rule:synchrony}
	add additional variables into the transformed \uppaal\ model.
	Hence, the mappings from
	$\{ \tranSet_i | \tranSet_i \in \statechart_i \wedge \statechart_i \in \statechartSet \}$
	and $\{ \varSet | \varSet \in \statechart_i \wedge \statechart_i \in \statechartSet \}$
	to $\{ \tranSet_i | \tranSet_i \in \automata_i \wedge \automata_i \in \automataSet \}$
	and $\{ \varSet | \varSet \in \automata_i \wedge \automata_i \in \automataSet \}$
	are not surjective, respectively.
	
	Therefore, the mappings from
	$\{ \tranSet_i | \tranSet_i \in \statechart_i \wedge \statechart_i \in \statechartSet \}$
	and $\{ \varSet | \varSet \in \statechart_i \wedge \statechart_i \in \statechartSet \}$
	to $\{ \tranSet_i | \tranSet_i \in \automata_i \wedge \automata_i \in \automataSet \}$
	and $\{ \varSet | \varSet \in \automata_i \wedge \automata_i \in \automataSet \}$
	are injective but not surjective, respectively.
\end{proof}

\begin{lemma}
	\label{lm:action}
	Given a basic \yakindu\ statechart 
	$\statechart = (\stateSet, \state_0, \tranSet, \guardSet, \actionSet, \varSet, \tranPrioritySet, \inStateActionAssign, \outStateActionAssign)$
	and its transformed \uppaal\ timed automaton
	$\automata = (\stateSet, \state_0, \tranSet, \guardSet, \actionSet, \varSet, \clockSet, \invariant)$
	by applying transformation Rule~\ref{rule:action},
	the statechart $\statechart$ and timed automaton $\automata$
	are equivalent.
\end{lemma}

\begin{proof}
	We assume that each state in the given \yakindu\ statechart $\statechart$
	only has one outgoing transition guarded by $\mathtt{true}$.
	The other two cases with non-$\mathtt{true}$ guards
	and multiple outgoing transitions will be proven in
	Lemmas~\ref{lm:guard}-\ref{lm:timer} and Lemma~\ref{lm:priority}, respectively.
	
	With the above assumptions, the semantics transition of the
	given \yakindu\ statechart $\statechart$, i.e.,
	formula~\eqref{eq:Ysenmantics1}, is simplified by substituting
	guards with $\mathtt{true}$ and omitting transition priority
	constraints as
	\begin{align}
	\label{eq:actionY}
	\begin{split}
	&(\state, \valueFun{\varSet}) \xrightarrow{\mathtt{true},<\outAction;\tranAction;\inAction>} (\state', \valueFun{\varSet}[<\outAction;\tranAction;\inAction>]) \\
	& \qquad \mathtt{ if } \ \exists (\state, \mathtt{ture}, \tranAction, \state') \in \tranSet :
	\valueFun{\varSet} \models \mathtt{ture}.
	\end{split}	
	\end{align}
	
	According to Lemma~\ref{lm:stMap}, the given \yakindu\ statechart $\statechart$
	has a unique transformed \uppaal\ timed automaton $\automata$.
	The semantics transition of the transformed \uppaal\ timed
	automaton $\automata$, i.e., formula~\eqref{eq:Usenmantics2},
	is simplified by applying transformation Rule~\ref{rule:action}
	and omitting clock constraints as
	\begin{align}
	\label{eq:actionU}
	\begin{split}
	&(\state, \valueFun{\varSet}) \xrightarrow{\mathtt{ture}, \action_{\u}} (\state', \valueFun{\varSet}[\action_{\u}]) \\
	&\qquad \mathtt{ if } \ \exists  (\state, \mathtt{ture}, \action_{\u}, \state') \in \tranSet : \valueFun{\varSet} \models \mathtt{ture},
	\end{split}	
	\end{align}	
	where $\action_{\u} = <\inState_{\tran_{\y}}(\outAction); \tran_{\y}(\action); \outState_{\tran_{\y}}(\inAction)>$.
	
	Based on Lemma~\ref{lm:stMap} and Lemma~\ref{lm:tranMap}, every state, transition, and variable in $\statechart$
	have an unique corresponding location, edge, and variable in $\automata$.
	According to the transformation Rule~\ref{rule:action} and
	the assumptions, the actions $<\outAction;\tranAction;\inAction>$ and
	$<\inState_{\tran_{\y}}(\outAction); \tran_{\y}(\action); \outState_{\tran_{\y}}(\inAction)>$ are equivalent.
	Hence, according to Rule~\ref{rule:action}, the semantics
	transitions of $\statechart$ and $\automata$, i.e.,
	formula~\eqref{eq:actionY} and formula~\eqref{eq:actionU},
	are equivalent. Therefore, $\statechart$ and $\automata$ have
	the equivalent execution traces under the same initial system
	statuses, which means $\statechart$ and $\automata$ are equivalent.
\end{proof}

\begin{lemma}
	\label{lm:guard}
	Given a basic \yakindu\ statechart 
	$\statechart = (\stateSet, \state_0, \tranSet, \guardSet, \actionSet, \varSet, \tranPrioritySet, \inStateActionAssign, \outStateActionAssign)$
	and its transformed \uppaal\ timed automaton
	$\automata = (\stateSet, \state_0, \tranSet, \guardSet, \actionSet, \varSet, \clockSet, \invariant)$
	by applying transformation Rule~\ref{rule:guard},
	the statechart $\statechart$ and timed automaton $\automata$
	are equivalent.
\end{lemma}

\begin{proof}
	We assume that each state in the given \yakindu\ statechart $\statechart$
	only has one outgoing transition, all transition guards do not
	contain events nor timing triggers,
	and all states and transitions in $\statechart$ do not contain actions.
	The case with actions have been proven in Lemma~\ref{lm:action}.
	The other three cases with events, timing triggers,
	and multiple outgoing transitions will be proven in
	Lemma~\ref{lm:event}, Lemma~\ref{lm:timer}, and Lemma~\ref{lm:priority},
	respectively.	
	
	With the above assumptions, the semantics transition of the
	given \yakindu\ statechart $\statechart$, i.e.,
	formula~\eqref{eq:Ysenmantics1}, is simplified by omitting actions
	and transition priority constraints as
	\begin{align}
	\label{eq:guardY}	
	(\state, \valueFun{\varSet}) \xrightarrow{\guard} (\state', \valueFun{\varSet})
	\ \ \mathtt{ if } \ \exists (\state, \guard, \state') \in \tranSet :
	\valueFun{\varSet} \models \guard.
	\end{align}
	
	According to Lemma~\ref{lm:stMap}, the given \yakindu\ statechart $\statechart$
	has a unique transformed \uppaal\ timed automaton $\automata$.
	The semantics transition of the transformed \uppaal\ timed
	automaton $\automata$, i.e., formula~\eqref{eq:Usenmantics2},
	is simplified by omitting actions and clock constraints as
	\begin{align}
	\label{eq:guardU}
	(\state, \valueFun{\varSet}) \xrightarrow{\guard} (\state', \valueFun{\varSet})
	\ \ \mathtt{ if } \ \exists  (\state, \guard, \state') \in \tranSet : \valueFun{\varSet} \models \guard.
	\end{align}
	
	Based on Lemma~\ref{lm:stMap} and Lemma~\ref{lm:tranMap}, every state, transition, and variable in $\statechart$
	have an unique corresponding location, edge, and variable in $\automata$.	
	The semantics transitions of $\statechart$ and $\automata$, i.e.,
	formula~\eqref{eq:guardY} and formula~\eqref{eq:guardU},
	are equivalent. Therefore, $\statechart$ and $\automata$ have
	the equivalent execution traces under the same initial system
	statuses, which means $\statechart$ and $\automata$ are equivalent.
\end{proof}

\begin{lemma}
	\label{lm:event}
	Given a basic \yakindu\ statechart 
	$\statechart = (\stateSet, \state_0, \tranSet, \guardSet, \actionSet, \varSet, \tranPrioritySet, \inStateActionAssign, \outStateActionAssign)$
	and its transformed \uppaal\ timed automaton
	$\automata = (\stateSet, \state_0, \tranSet, \guardSet, \actionSet, \varSet, \clockSet, \invariant)$
	by applying transformation Rule~\ref{rule:event},
	the statechart $\statechart$ and timed automaton $\automata$
	are equivalent.
\end{lemma}

\begin{proof}
	We assume that each state in the given \yakindu\ statechart $\statechart$
	only has one outgoing transition, 
	the guard of every transition is an \textit{event} $\event$,
	and all states and transitions in $\statechart$ do not contain actions.
	The cases with actions and regular guards have been proven in Lemma~\ref{lm:action} and Lemma~\ref{lm:guard}, respectively.
	The other cases with timing triggers and multiple outgoing transitions will be proven in Lemma~\ref{lm:timer} and
	Lemma~\ref{lm:priority}, respectively.
	
	With the above assumptions, the semantics transition of the
	given \yakindu\ statechart $\statechart$, i.e.,
	formula~\eqref{eq:Ysenmantics1}, is simplified by substituting
	guards with $\event$ and omitting actions and transition priority constraints as	
	\begin{align}
	\label{eq:eventY}
	(\state, \valueFun{\varSet}) \xrightarrow{\event} (\state', \valueFun{\varSet})
	\quad \mathtt{ if } \ \exists (\state, \event, \state') \in \tranSet :
	\valueFun{\varSet} \models \event.
	\end{align}
	
	According to Lemma~\ref{lm:stMap}, the given \yakindu\ statechart $\statechart$
	has a unique transformed \uppaal\ timed automaton $\automata$.
	The semantics transition with channel synchronizations of the
	transformed \uppaal\ timed automaton $\automata$, i.e.,
	formula~\eqref{eq:Usenmantics5}, is simplified by applying
	Rule~\ref{rule:event} and omitting actions and clock constraints as	
	\begin{align}
	\label{eq:eventU}
	\begin{split}
	&(\overline{\state}, \valueFun{\varSet}) \xrightarrow{\event}
	(\overline{\state}[\state_i^{'}/\state_i, \state_j^{'}/\state_j], \valueFun{\varSet})\\
	&\quad \mathtt{ if } \ \exists 
	(\state_i, \event?, \state'_i) \in \tranSet_i \wedge 
	(\state_j, \event!, \state'_j) \in \tranSet_j \ :
	\valueFun{\varSet} \models \event.
	\end{split}			
	\end{align}
	In formula~\eqref{eq:eventU}, $\state_i$ and $\tranSet_i$ denote
	the location and the edge set of the transformed \uppaal\
	timed automaton; $\state_j$ and $\tranSet_j$ denote
	the location and the edge set of the auxiliary event
	automaton.
	
	Based on Lemma~\ref{lm:stMap} and Lemma~\ref{lm:tranMap}, every state, transition, and variable in $\statechart$
	have an unique corresponding location, edge, and variable in $\automata$.
	According to semantics transitions of $\statechart$ and $\automata$, i.e.,
	formula~\eqref{eq:eventY} and formula~\eqref{eq:eventU},
	if the event $\event$ is triggered, $\statechart$ and
	$\automata$ have the equivalent execution traces under
	the same initial system statuses; otherwise, both $\statechart$
	and $\automata$ stay in their current state.
	Therefore, the statechart $\statechart$ and
	timed automaton $\automata$ are equivalent.
\end{proof}

\begin{lemma}
	\label{lm:timer}
	Given a basic \yakindu\ statechart 
	$\statechart = (\stateSet, \state_0, \tranSet, \guardSet, \actionSet, \varSet, \tranPrioritySet, \inStateActionAssign, \outStateActionAssign)$
	and its transformed \uppaal\ timed automaton
	$\automata = (\stateSet, \state_0, \tranSet, \guardSet, \actionSet, \varSet, \clockSet, \invariant)$
	by applying transformation Rule~\ref{rule:timer},
	the statechart $\statechart$ and timed automaton $\automata$
	are equivalent.
\end{lemma}

\begin{proof}	
	We assume that each state in the given \yakindu\ statechart $\statechart$
	only has one outgoing transition, 
	the guard of every transition is a \textit{timing trigger} $\timer$,
	and all states and transitions in $\statechart$ do not contain actions.
	The cases with actions, regular guards, and event guards have been proven in Lemma~\ref{lm:action}, Lemma~\ref{lm:guard}, and Lemma~\ref{lm:event}, respectively.
	The other case with multiple outgoing transitions will be proven in
	Lemma~\ref{lm:priority}.
	
	With the above assumptions, the semantics transition of the
	given \yakindu\ statechart $\statechart$, i.e.,
	formula~\eqref{eq:Ysenmantics1}, is simplified by substituting
	guards with $\timer$ and omitting actions and transition priority constraints as		
	\begin{align}
	\label{eq:timerY}
	(\state, \valueFun{\varSet}) \xrightarrow{\timer} (\state', \valueFun{\varSet})
	\ \ \mathtt{ if } \ \exists (\state, \timer, \state') \in \tranSet :
	\valueFun{\varSet} \models \timer.
	\end{align}
	
	According to Lemma~\ref{lm:stMap}, the given \yakindu\ statechart $\statechart$
	has a unique transformed \uppaal\ timed automaton $\automata$.
	The semantics transition with channel synchronizations of the
	transformed \uppaal\ timed automaton $\automata$, i.e.,
	formula~\eqref{eq:Usenmantics5}, is simplified by applying
	Rule~\ref{rule:timer} and omitting actions and clock constraints as	
	\begin{align}
	\label{eq:timerU}
	\begin{split}
	&(\overline{\state}, \valueFun{\varSet}) \xrightarrow{\timer}
	(\overline{\state}[\state_i^{'}/\state_i, \state_j^{'}/\state_j], \valueFun{\varSet})\\
	&\ \mathtt{ if } \ \exists 
	(\state_i, \var_{\timer}?, \state'_i) \in \tranSet_i \wedge 
	(\state_j, \var_{\timer}!, \state'_j) \in \tranSet_j \ :
	\valueFun{\varSet} \models \timer.
	\end{split}			
	\end{align}
	According to Rule~\ref{rule:timer}, all clock constraints are
	implemented in the auxiliary timing trigger automaton.
	The transformed timed automaton $\automata$ only synchronizes
	with the auxiliary timing trigger automaton through
	channel $\var_{\timer}$. As we focus on the execution behavior
	of the transformed timed automaton $\automata$, hence
	the semantics transition of $\automata$, i.e.,
	formula~\eqref{eq:timerU}, can be simplified by omitting
	clock constraints.	
	In formula~\eqref{eq:timerU}, $\state_i$ and $\tranSet_i$ denote
	the location and the edge set of the transformed \uppaal\
	timed automaton; $\state_j$ and $\tranSet_j$ denote
	the location and the edge set of the auxiliary timing trigger
	automaton.
	
	Based on Lemma~\ref{lm:stMap} and Lemma~\ref{lm:tranMap}, every state, transition, and variable in $\statechart$
	have an unique corresponding location, edge, and variable in $\automata$.	
	According to semantics transition of $\statechart$ and $\automata$, i.e.,
	formula~\eqref{eq:timerY} and formula~\eqref{eq:timerU},
	if the timing trigger $\timer$ is enabled, $\statechart$ and
	$\automata$ have the equivalent execution traces under
	the same initial system statuses; otherwise, both $\statechart$
	and $\automata$ stay in their current state.
	Therefore, the statechart $\statechart$ and
	timed automaton $\automata$ are equivalent.	
\end{proof}

\begin{lemma}
	\label{lm:priority}
	Given a basic \yakindu\ statechart 
	$\statechart = (\stateSet, \state_0, \tranSet, \guardSet, \actionSet, \varSet, \tranPrioritySet, \inStateActionAssign, \outStateActionAssign)$
	and its transformed \uppaal\ timed automaton
	$\automata = (\stateSet, \state_0, \tranSet, \guardSet, \actionSet, \varSet, \clockSet, \invariant)$
	by applying transformation Rule~\ref{rule:priority},
	the statechart $\statechart$ and timed automaton $\automata$
	are equivalent.
\end{lemma}

\begin{proof}
	We assume that all states and transitions in $\statechart$ do not contain actions.
	The case with actions has been proven in Lemma~\ref{lm:action}.
	
	With the above assumptions, the semantics transition of the
	given \yakindu\ statechart $\statechart$, i.e.,
	formula~\eqref{eq:Ysenmantics1}, is simplified by omitting actions
	as
	\begin{align}
	\label{eq:priorityY}
	\begin{split}
	&(\state, \valueFun{\varSet}) \xrightarrow{\guard, \tranPriority} (\state', \valueFun{\varSet}) \\
	&\ \ \mathtt{ if } \ \exists (\state, \guard, \tranPriority, \state') \in \tranSet :
	(\valueFun{\varSet} \models \guard) \wedge \prifun(\state, \state', \tranPriority, \valueFun{\varSet}),
	\end{split}	
	\end{align}
	where $\prifun(\state, \state', \tranPriority, \valueFun{\varSet}) \equiv \forall (\state, \guard', *, \tranPriority', *) \in \tranSet : 
	\valueFun{\varSet} \models \guard' \wedge \tranPriority \le \tranPriority'$.
	
	According to Lemma~\ref{lm:stMap}, the given \yakindu\ statechart $\statechart$
	has a unique transformed \uppaal\ timed automaton $\automata$.
	The semantics transition of the
	transformed \uppaal\ timed automaton $\automata$, i.e.,
	formula~\eqref{eq:Usenmantics2}, is simplified by applying
	Rule~\ref{rule:priority} and omitting actions and clock constraints as	
	\begin{align}
	\label{eq:priorityU}
	\begin{split}
	&(\state, \valueFun{\varSet}) \xrightarrow{\guard \bigwedge \neg \guard' } (\state', \valueFun{\varSet}) \\
	&\qquad \mathtt{ if } \ \exists  (\state, \guard \bigwedge \neg \guard', \state') \in \tranSet : \valueFun{\varSet} \models \guard \bigwedge \neg \guard',
	\end{split}	
	\end{align}
	where $\guard'$ is the guard of transition in $\{ \tran | \tran \in (\state, \guard', *, \tranPriority', *) \wedge \tranPriority > \tranPriority' \}$.
	
	Based on Lemma~\ref{lm:stMap} and Lemma~\ref{lm:tranMap}, every state, transition, and variable in $\statechart$
	have an unique corresponding location, edge, and variable in $\automata$.
	In formula~\eqref{eq:priorityU}, the guard $\guard'$ satisfies
	$\forall (\state, \guard', *, \tranPriority', *) \in \tranSet: \tranPriority > \tranPriority'$.
	Hence, $\neg \guard'$ satisfies $\forall (\state, \guard', *, \tranPriority', *) \in \tranSet: \tranPriority \le \tranPriority'$,
	which can be implied by $\prifun(\state, \state', \tranPriority, \valueFun{\varSet})$ in formula~\eqref{eq:priorityY}.
	Therefore, $\statechart$ and $\automata$ have
	the equivalent execution traces under the same initial system
	statuses, which means $\statechart$ and $\automata$ are equivalent.	
\end{proof}

\begin{lemma}
	\label{lm:synchrony}	
	Given a \yakindu\ statechart model
	$\statechartSet=\{ (\statechart_i, \statechartPriority_i) | \statechart_i = (\stateSet_i, \state^0_i, \tranSet_i, \guardSet_i, \actionSet_i, \varSet, \tranPrioritySet_i, \inStateActionAssign^i, \outStateActionAssign^i) \wedge \statechartPriority_i = i \wedge 1 \le i \le n \}$
	and its transformed \uppaal\ timed automata model
	$\automataSet = \{ \automata_i | \automata_i = (\stateSet_i, \state^0_i, \tranSet_i, \guardSet_i, \actionSet_i, \varSet, \clockSet, \invariant_i) \wedge 1 \le i \le n \}$ by applying transformation Rule~\ref{rule:synchrony},
	the statechart model $\statechartSet$ and timed automaton model $\automataSet$ are equivalent.
\end{lemma}

\begin{proof}	
	We assume that each state in the given \yakindu\ statechart $\statechart$
	only has one outgoing transition and 
	all states and transitions in $\statechart$ do not contain actions.
	The cases with actions and multiple outgoing transitions
	have been proven in Lemma~\ref{lm:action} and Lemma~\ref{lm:priority}, respectively.
	
	With the above assumptions, the semantics transition of the
	given \yakindu\ statechart $\statechart$, i.e.,
	formula~\eqref{eq:Ysenmantics2-2}, is simplified by
	omitting actions and transition priority constraints as
	\begin{align}
	\label{eq:synchronyY}
	\begin{split}
	&(\overline{\state},\valueFun{\varSet})
	\xrightarrow{\guard,\statechartPriority_i}
	(\overline{\state}[\state'_i/\state_i], \valueFun{\varSet}[\increaseStep])\\
	&\qquad \mathtt{if} \ \statechartPriority_i==\step \ \wedge \
	\exists (\state_i, \guard, \state'_i) \in \tranSet_i :
	\valueFun{\varSet} \models \guard.
	\end{split}			
	\end{align}
	
	According to Lemma~\ref{lm:stMap}, the given \yakindu\ statechart $\statechart$
	has a unique transformed \uppaal\ timed automaton $\automata$.
	The semantics transition of the
	transformed \uppaal\ timed automaton $\automata$, i.e.,
	formula~\eqref{eq:Usenmantics4}, is simplified by applying
	Rule~\ref{rule:synchrony} and omitting actions and clock constraints as
	\begin{align}
	\label{eq:synchronyU}
	\begin{split}
	&(\overline{\state}, \valueFun{\varSet}) \xrightarrow{\guard \ \&\& \ (\step == \statechartPriority_i)} (\overline{\state}[\state_i^{'}/\state_i], \valueFun{\varSet}[\increaseStep]) \\
	&\qquad \mathtt{ if } \ \exists (\state_i, \guard \ \&\& \ (\step == \statechartPriority_i), \state'_i) \in \tranSet_i \ :
	\valueFun{\varSet} \models \guard \ \&\& \ \step == \statechartPriority_i
	\end{split}			
	\end{align}
	
	Based on Lemma~\ref{lm:stMap} and Lemma~\ref{lm:tranMap}, every state, transition, and variable in $\statechart$
	have an unique corresponding location, edge, and variable in $\automata$.
	The semantics
	transitions of $\statechartSet$ and $\automataSet$, i.e.,
	formula~\eqref{eq:synchronyY} and formula~\eqref{eq:synchronyU},
	are equivalent. Therefore, $\statechartSet$ and $\automataSet$ have
	the equivalent execution traces under the same initial system
	statuses, which means $\statechartSet$ and $\automataSet$ are equivalent.
\end{proof}

\begin{theorem}
	\label{thm:singlest}
	Given a basic \yakindu\ statechart $\statechart$
	and its transformed \uppaal\ timed automaton $\automata$,	
	the statechart $\statechart$ and timed automaton $\automata$
	are equivalent.
\end{theorem}

\begin{proof}	
	We prove the theorem by induction on the number $N$ of iterations the rule is applied.\\
	\noindent \textbf{Base Case:} When $N=1$, prove that the transformation maintains the execution semantics equivalence.
	To transform the given \yakindu\ statechart $\statechart$, 
	Rules~\ref{rule:initial}-\ref{rule:priority}
	can be applied. We have proven that Rule~\ref{rule:initial}
	maintains elements' one-to-one mapping through Lemma~\ref{lm:stMap}
	and Rules~\ref{rule:action}-\ref{rule:priority} maintain model
	equivalence through Lemma~\ref{lm:action} to Lemma~\ref{lm:priority}.
	Hence, the statement holds for the base case.\\
	\noindent \textbf{Induction Step:} Assume the statement is true when $N=k$ and prove that
	it also holds when $N=k+1$. According to the transformation rules, each
	rule does not interfere the execution semantics of other rules.
	The base case has proven each transformation rule maintains the 
	execution equivalence.
	Therefore, if the statement holds when $N=k$, then it is also true when $N=k+1$.
\end{proof}

\begin{theorem}
	\label{thm:mulst}
	Given a \yakindu\ statechart model $\statechartSet$
	and its transformed \uppaal\ timed automata model $\automataSet$,
	the statechart model $\statechartSet$ and timed automata model
	$\automataSet$ are equivalent.
\end{theorem}

\begin{proof}
	If the given \yakindu\ model $\statechartSet$ only contains one statechart,
	the statement holds as proven in Theorem~\ref{thm:singlest}.
	If $\statechartSet$ contains multiple statecharts, we prove the theorem with
	similar induction approach as the proof of Theorem~\ref{thm:singlest}.
	Compared with Theorem~\ref{thm:singlest}, the only difference is that the transformation
	of multiple statecharts needs to apply Rule~\ref{rule:synchrony} which has been proven
	in Lemma~\ref{lm:synchrony}. Therefore, $\statechartSet$ and $\automataSet$ are equivalent.
\end{proof}

\section{Simplified Cardiac Arrest Treatment Case Study}
\label{sec:exp}
\subsection{Simplified Cardiac Arrest Treatment Scenario}
\label{subsec:cardiacExp}
Cardiac arrest is the abrupt loss of heart function and
can lead to death within minutes.
In a simplified cardiac arrest treatment scenario~\cite{WuTreatment2014}, medical staff intend
to activate a defibrillator to deliver a therapeutic level of
electrical shock that can correct certain types of deadly irregular
heart-beats such as ventricular fibrillation. The medical
staff need to check two preconditions: (1) patient's airway and
breathing are under control and (2) the electrocardiogram (ECG) monitor shows a
shockable rhythm.
Suppose the patient's airway is open and
breathing is under control, but the ECG monitor shows a
non-shockable rhythm.
In order to induce a shockable rhythm,
a drug, called epinephrine (EPI), is commonly given to increase
cardiac output. Giving epinephrine, however, also has two
preconditions: (1) patient's blood pH value should be larger than
7.4 and (2) urine flow rate should be greater than 12 mL/s.
In order to correct these two preconditions, sodium bicarbonate
should be given to raise blood pH value, and intravenous (IV) fluid
should be increased to improve urine flow rate.

The simplified cardiac arrest treatment scenario has two safety properties:
(1) \textbf{P1:} the defibrillator is activated only if the ECG rhythm is
shockable and breathing is normal; and
(2) \textbf{P2:} the epinephrine is injected only if the blood pH
value is larger than 7.4 and urine flow rate is higher than 12 mL/s.

\subsection{Model Validation and Verification}
\label{subsec:cardiacModel}

Wu~\etal\ developed a validation protocol to enforce the correct execution
sequence of performing treatment, regarding preconditions validation, side
effects monitoring, and expected responses checking based on
the pathophysiological models~\cite{WuTreatment2014}.
We use \yakindu\ statecharts~\cite{yakindu} to model the simplified
cardiac arrest treatment procedure with the validation
protocol~\cite{WuTreatment2014}, as depicted in Fig.~\ref{fig:cardiacY}.
The statechart model consists of six statecharts: $\mathtt{Treatment}$,
$\mathtt{Ventilator}$, $\mathtt{EPIpump}$, $\mathtt{SodiumBicarbonatePump}$,
$\mathtt{IVpump}$, and $\mathtt{LasixPump}$.
The statecharts communicate using events and shared variables.
The $\mathtt{Treatment}$ statechart implements the simplified
cardiac arrest treatment procedure.
The other statecharts implement treatment actions, such as medicine injection.
\begin{figure*}[ht]
	\centering
	\includegraphics[width =0.99\textwidth]{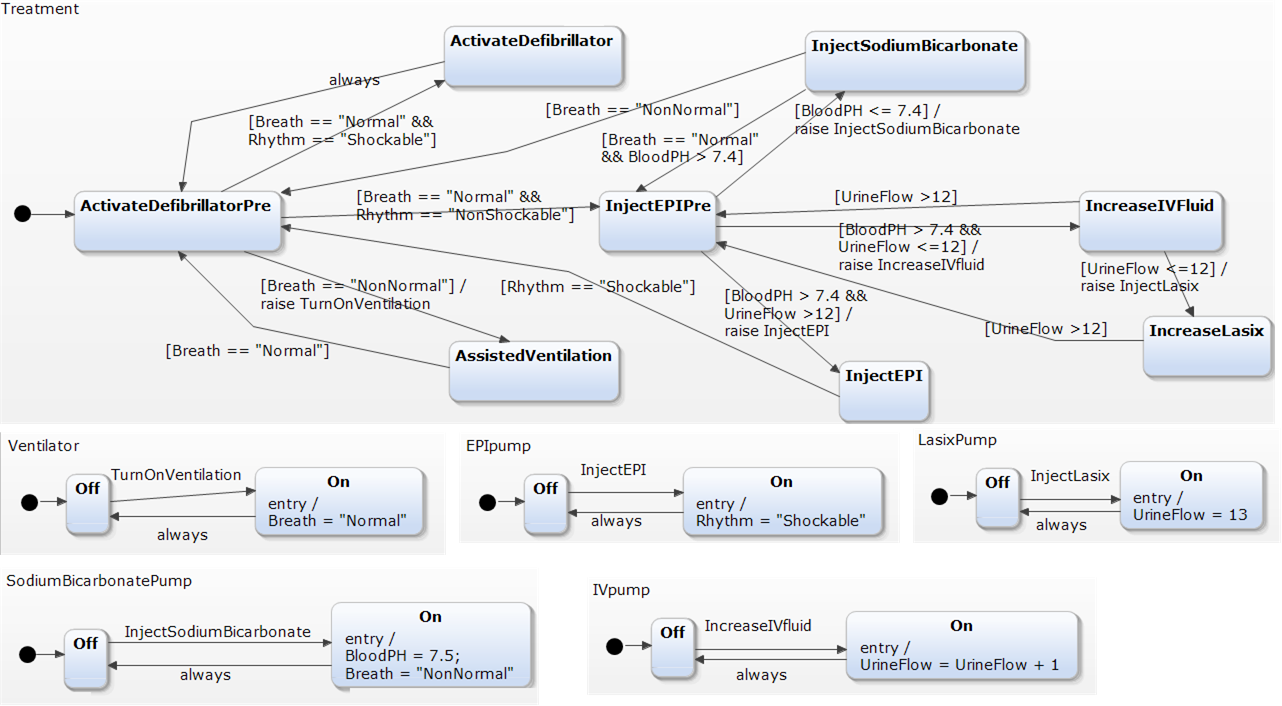}
	\caption{Simplified Cardiac Arrest Treatment Statechart Model}
	\label{fig:cardiacY}
\end{figure*}

We use the \toolname\footnote{The \toolname\ tool is available in \toolweb.} tool to transform the simplified cardiac arrest
treatment statechart model given in Fig.~\ref{fig:cardiacY}
to timed automata to formally verify safety properties \textbf{P1} and \textbf{P2}.
The transformed simplified cardiac arrest treatment timed automata
model is depicted in Fig.~\ref{fig:cardiacU}.
\begin{figure*}[ht]
	\centering
	\includegraphics[width = 0.99\textwidth]{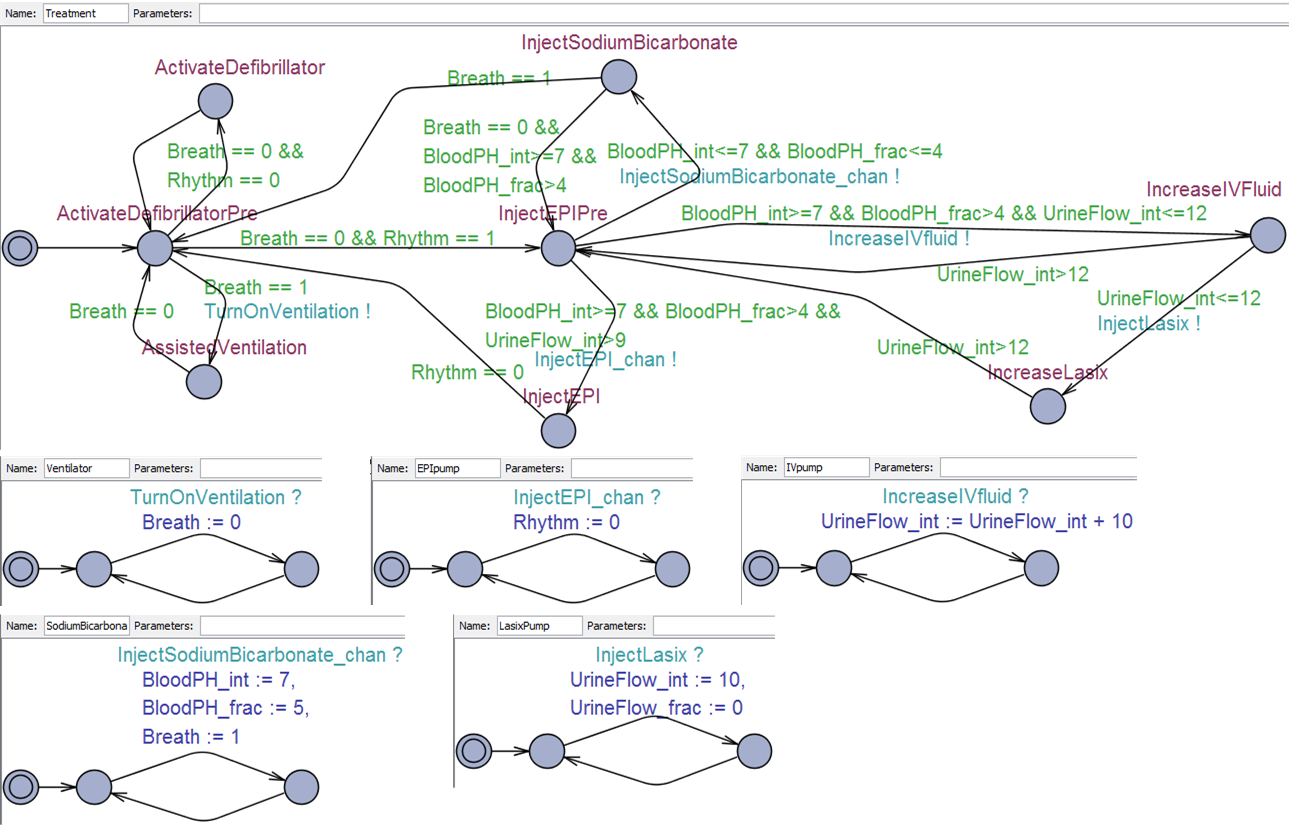}
	\caption{Simplified Cardiac Arrest Treatment Timed Automata Model}
	\label{fig:cardiacU}
\end{figure*}

The safety properties \textbf{P1} and \textbf{P2} are specified by CTL (computation tree logic)
formula~\eqref{eq:P1} and formula~\eqref{eq:P2}, respectively.
\begin{align}
\label{eq:P1}
\begin{split}
&A[~] \ \mathtt{Treatment.ActivateDefibrillaotr} \ imply \\
&\qquad \mathtt{Breath} == 0 \ \&\& \ \mathtt{Rhythm} == 0
\end{split}
\end{align}
\begin{align}
\label{eq:P2}
\begin{split}
&A[~] \ \mathtt{Treatment.InjectEPI} \ imply \\
&\qquad (\mathtt{BloodPH_{int}}>7 \ || \ (\mathtt{BloodPH_{int}}==7) \ \&\& \\
&\qquad \mathtt{BloodPH_{frac}}>4) \ \&\& \ (\mathtt{UrineFlow_{int}}>12 \ || \\
&\qquad (\mathtt{UrineFlow_{int}}==12 \ \&\& \ \mathtt{UrineFlow_{frac}}>0) )
\end{split}
\end{align}
Note that \uppaal\ timed automata do not support real numbers and strings.
For a string, we represent it by an integer variable with a dictionary
which maps integer values to string values. For instance, in formula~\eqref{eq:P1},
$\mathtt{Breath} == 0$ and $\mathtt{Rhythm} == 0$ indicate
the patient breathing is normal and the ECG rhythm is shockable, respectively.
To represent a real number in \uppaal\ timed automata, we use two integers
to represent its integer part and fraction part, respectively.
For example, in formula~\eqref{eq:P2},
$\mathtt{BloodPH_{int}}$ and $\mathtt{BloodPH_{frac}}$ represent
the integer and fraction parts of the patient's blood pH value, respectively.

We first run simulations on the simplified cardiac arrest treatment statecharts
through \yakindu\ to validate safety properties \textbf{P1} and \textbf{P2}.
The simulation results through \yakindu\ show that both safety properties
\textbf{P1} and \textbf{P2} are satisfied.
We then verify the two safety properties represented by formula~\eqref{eq:P1}
and formula~\eqref{eq:P2} in \uppaal.
The verification results also show that both \textbf{P1} and \textbf{P2} hold.

To validate the proposed approach,
we inject an error into the simplified cardiac arrest treatment statechart model
shown in Fig.~\ref{fig:cardiacY} as follows:
change the guard of transition, which transits from
state $\mathtt{InjectEPIPre}$ to state $\mathtt{InjectEPI}$,
from $\mathtt{BloodPH} > 7.4 \ \&\& \ \mathtt{UrineFlow} >12$
to $\mathtt{BloodPH} > 7.4 \ \&\& \ \mathtt{UrineFlow} >10$.
The injected error should fail the safety property \textbf{P2}.
We re-transform the statechart model with the injected error to
timed automata and verify the safety properties \textbf{P1} and \textbf{P2}
through \uppaal. The verification results show that the safety property
\textbf{P1} still holds, while the safety property \textbf{P2} fails.

%
%

\section{Related Work}
\label{sec:related}
Medical best practice guidelines play an important role in today's
medical care. There exist many efforts to develop various computer-interpretable
models and tools for the management of guidelines,
such as GLIF~\cite{patel1998representing}, Asbru~\cite{Balser2002Asbru},
EON~\cite{Tu2001EON}, GLARE~\cite{Terenziani2004GLARE},
and PROforma~\cite{Fox1998PROforma}, which are mainly aimed to provide guided support
to doctors. The exist medical guideline modeling approaches can improve effectiveness
of clinical validation. However for safety-critical medical guideline systems,
validation by medical staff alone is not adequate for guaranteeing
safety, formal verifications are needed. The formal model based approaches~\cite{Clarke1999ModelCheckingBook,Loveland1978TheoremProvingBook,Duftschmid2001KVC}
are appealing because they provide a unified basis for formal analysis to achieve the
expected level of safety guarantees.
Unfortunately, most existing medical guideline models,
such as Asbru~\cite{Balser2002Asbru} and GLARE~\cite{Terenziani2004GLARE},
do not provide formal verification capability.
A common approach is to transform
an existing medical guideline model to a formal model to verify
the safety properties, such as transforming Asbru model to KIV model~\cite{Hommersom2007TKDE}
and transforming GLARE model to PROMELA model~\cite{Giordan2006}.

To bridge the gap between state-oriented models and formal verification,
efforts are also made from research community to transform
state-oriented modeling specifications/languages, 
such as UML (unified modeling language) statecharts~\cite{nobakht2013approach,Zorin2012Translation},
hierarchical timed automata (HTA)~\cite{David2002FASE}, discrete
event system specification for real-time (RT-DEVS)~\cite{furfaro2009development},
parallel object-oriented specification language (POOSL)~\cite{Xing2010ACSD},
and Stateflow models~\cite{Jiang2016RTAS,Jiang2017TCPS}
to \uppaal\ timed automata.
On the other hand, Pajic \etal\ developed a tool to transform
\uppaal\ timed automata to Stateflow models for implementation issues~\cite{Pajic2012RTAS,Pajic2014TECS}.

\yakindu\ statecharts are similar to UML statecharts, but
have major semantic differences~\cite{yakinduDoc}. For example,
the execution semantics of \yakindu\ statecharts is cycle driven,
while UML statecharts are event driven. Therefore,
existing tools cannot be directly applied to transform
\yakindu\ statecharts to \uppaal\ timed automata. 
Moreover, most existing work do not provide formal definitions for
transformation nor formally prove the correctness of transformation.
The work presented in this paper overcomes the limitations rendered
by statecharts, provide formalized transformation rules, and
prove that the transformation rules maintain model equivalence.

\section{Conclusion}
\label{sec:conclusion}
The existing medical best practice guidelines in hospital handbooks
are often lengthy and difficult for medical staff to 
remember and apply clinically.
Our previous work~\cite{Guo2016ICCPS} designed and implemented a framework
to support developing verifiably safe medical guideline models.
The framework models medical guidelines with statecharts and
transform statecharts to timed automata to formally verify
safety properties.
However, some components in the framework do not have formal
definitions, which will result in unavoidable ambiguity.
The paper presents the formalism of the framework in four folds:
(1) define the formal execution semantics of \yakindu\ statecharts;
(2) formalize the transformation rules from statecharts to
timed automata;
(3) give the formal definition of execution semantics equivalence
between statecharts and transformed timed automata; and
(4) formally prove that the transformation rules maintain the execution semantic
equivalence between statecharts and transformed timed automata.
We take \yakindu\ and \uppaal\ as examples to implement the framework.
The methodology can be applied to other modeling languages/tools and
other safety-critical domains.

\begin{acks}
	This work is supported in part by NSF CNS 1842710,
NSF CNS 1545008, and NSF CNS 1545002.
\end{acks}

\bibliographystyle{ACM-Reference-Format}
\bibliography{ref}

\end{document}